\documentclass{LMCS}

\def\doi{8 (1:12) 2012}
\lmcsheading%
{\doi}
{1--38}
{}
{}
{May\phantom.~26, 2011}
{Feb.~27, 2012}
{}

\usepackage{natbib}
\usepackage[utf8]{inputenc}
\usepackage[T1]{fontenc}
\usepackage{ae,aecompl} 

\usepackage{amsmath, amsfonts}
\usepackage{graphicx}
\usepackage{txfonts}
\usepackage{url}
\usepackage{epsfig}
\usepackage{subfigure}

\usepackage{enumerate}
\usepackage{hyperref}

\usepackage{algorithm}
\usepackage{algorithmic}

\usepackage[mathscr]{euscript}

\newcommand{\fancypicture}[1]{}

\renewcommand{\cite}{\citet}
\newcommand{\rightqed}{}

\newcommand{\define}[1]{\emph{#1}}
\newcommand{\tuple}[1]{\langle{#1}\rangle}
\newcommand{\defeq}{\coloneqq}

\newcommand{\EL}{$\mathcal{EL}$}
\newcommand{\SHIQ}{$\mathcal{SHIQ}$}
\newcommand{\SHIQb}{$\mathcal{SHIQ}b_s$}
\newcommand{\SHIQbs}{$\mathcal{SHIQ}b_s$}
\newcommand{\SRIQbs}{$\mathcal{SRIQ}b_s$}
\newcommand{\SHOIQ}{$\mathcal{SHOIQ}$}
\newcommand{\SROIQ}{$\mathcal{SROIQ}$}
\newcommand{\SHOIQBs}{$\mathcal{SHOIQB}_s$}
\newcommand{\SROIQBs}{$\mathcal{SROIQB}_s$}
\newcommand{\ALC}{$\mathcal{ALC}$}

\newcommand{\ALCHIQb}{$\mathcal{ALCHIQ}b$}
\newcommand{\ALCIQb}{$\mathcal{ALCQI}b$}
\newcommand{\ALCHIqb}{$\mathcal{ALCHI}b^{\leqslant}$}
\newcommand{\ALCIqb}{$\mathcal{ALCI}b^{\leqslant}$}
\newcommand{\ALCIFb}{$\mathcal{ALCIF}b$}
\newcommand{\ALCIb}{$\mathcal{ALCI}b$}

\newcommand{\NExpTime}{{\sc{NExpTime}}}
\newcommand{\NExpExpTime}{{\sc{N2ExpTime}}}
\newcommand{\ExpTime}{{\sc{ExpTime}}}

\newcommand{\NP}{{\sc{NP}}}

\newcommand{\Inter}{\mathcal{I}} 
\newcommand{\Jnter}{\mathcal{J}} 
\newcommand{\Knter}{\mathcal{K}} 
\newcommand{\true}{\mathit{true}}
\newcommand{\false}{\mathit{false}}

\newcommand{\delem}{\delta} 
\newcommand{\selem}{\sigma} 

\RequirePackage{graphicx}
\newcommand{\quantor}{\mathord{\reflectbox{$\text{\sf{Q}}$}}}

\newcommand{\kb}{\text{\rm{KB}}} 
\newcommand{\rb}{\text{\rm{RB}}} 
\newcommand{\kbrb}{\kb} 
\newcommand{\lang}[1]{\ensuremath{\mathbf{#1}}} 
\newcommand{\rolnames}{{\text{\sf{N}}_R}} 
\newcommand{\connames}{{\text{\sf{N}}_C}} 
\newcommand{\indnames}{{\text{\sf{N}}_I}} 
\newcommand{\Rlang}{\ensuremath{\mathbf{R}}} 
\newcommand{\Var}{\text{\sf{Var}}} 

\newcommand{\ssb}{\sqsubseteq}
\newcommand{\atleast}[1]{\mathord{\geqslant}#1\,}
\newcommand{\atmost}[1]{\mathord{\leqslant}#1\,}

\newcommand{\con}[1]{#1}
\newcommand{\rol}[1]{#1}

\newcommand{\rolR}{\rol{R}}
\newcommand{\rolS}{\rol{S}}
\newcommand{\rolT}{\rol{T}}
\newcommand{\rolU}{\rol{U}}
\newcommand{\rolV}{\rol{V}}
\newcommand{\rolW}{\rol{W}}
\newcommand{\conA}{\con{A}}

\newcommand{\conC}{\con{C}}
\newcommand{\conD}{\con{D}}
\newcommand{\conE}{\con{E}}
\newcommand{\conF}{\con{F}}
\newcommand{\dlname}[1]{\text{\it{#1}}} 

\newcommand{\Trans}{\text{\sf{Tra}}} 
\newcommand{\Inv}{\text{\rm{Inv}}} 

\newcommand{\Epsilon}{\Theta}
\newcommand{\NNF}{\text{\sf{NNF}}} 
\newcommand{\FLAT}{\text{{\sf FLAT}}} 
\newcommand{\Es}{\Epsilon_{\mathcal{S}}} 
\newcommand{\Ege}{\Epsilon_{\geqslant}} 
\newcommand{\Eh}{\Epsilon_{\mathcal{H}}} 
\newcommand{\Ele}{\Epsilon_{\leqslant}} 
\newcommand{\Ef}{\Epsilon_{\mathcal{F}}} 
\newcommand{\Eshq}{\Epsilon_{\mathcal{SHQ}}} 

\newcommand{\domp}{\pi} 
\newcommand{\DO}{\mathbb{D}} 
\newcommand{\tail}[1]{\mbox{\small\textsf{last}}(#1)} 
\newcommand{\val}{\text{\textrm{tail}}} 
\newcommand{\parts}{P} 
\newcommand{\roleval}[2]{#1\vdash #2} 
\newcommand{\nroleval}[2]{#1\not\vdash #2}

\newcommand{\charf}[1]{\llbracket{}#1\rrbracket{}_{\chi}} 
\newcommand{\tobf}[1]{\llbracket{}#1\rrbracket{}} 

\newcommand{\high}{\text{\sf high}}
\newcommand{\low}{\text{\sf low}}
\newcommand{\nroot}{n_{\:\!\text{\sf root}}}
\newcommand{\ntrue}{n_{\;\!\text{\sf true}}}
\newcommand{\nfalse}{n_{\;\!\text{\sf false}}}

\newcommand{\Prog}{\mathsf{DD}}

\theoremstyle{plain}
\newtheorem{theorem}{Theorem}[section]
\newtheorem{proposition}[theorem]{Proposition}
\newtheorem{lemma}[theorem]{Lemma}
\theoremstyle{definition}
\newtheorem{definition}[theorem]{Definition}

\begin{document}
\setlength{\parindent}{0mm}

\title[Type-Elimination-Based Reasoning for \SHIQb{}]{Type-Elimination-Based Reasoning for the Description Logic \SHIQb{} Using Decision Diagrams and Disjunctive
Datalog}

\author[S.~Rudolph]{Sebastian Rudolph\rsuper a}   
\address{{\lsuper a}Institute AIFB, Karlsruhe Institute of Technology, Germany} 
\email{rudolph@kit.edu}  

\author[M.~Kr\"{o}tzsch]{Markus Kr\"{o}tzsch\rsuper b} 
\address{{\lsuper b}Department of Computer Science, University of Oxford, UK}    
\email{markus.kroetzsch@cs.ox.ac.uk}  

\author[P.~Hitzler]{Pascal Hitzler\rsuper c} 
\address{{\lsuper c}Kno.e.sis, Wright State University, Dayton, Ohio, US} 
\email{pascal.hitzler@wright.edu}  





\begin{abstract}
\noindent We propose a novel, type-elimination-based method for
standard reasoning in the description logic \SHIQb{} extended by
DL-safe rules. To this end, we first establish a knowledge
compilation method converting the terminological part of an \ALCIb{}
knowledge base into an ordered binary decision diagram (OBDD) that
represents a canonical model. This OBDD can in turn be transformed
into disjunctive Datalog and merged with the assertional part of the
knowledge base in order to perform combined reasoning.
In order to leverage our technique for full \SHIQb{}, we provide a
stepwise reduction from \SHIQb{} to \ALCIb{} that preserves
satisfiability and entailment of positive and negative ground facts.
The proposed technique is shown to be worst-case optimal w.r.t.\
combined and data complexity.
\end{abstract}

\maketitle


\section{Introduction}\label{sec:intro}

Description logics (DLs, see \citealp{dlhandbook}) have become a major paradigm in Knowledge Representation and Reasoning. This can in part be attributed to the fact that the DLs have been found suitable to be the foundation for ontology modeling and reasoning for the Semantic Web. In particular, the Web Ontology Language OWL \citep{owl2-overview}, a recommended standard by the World Wide Web Consortium (W3C)\footnote{http://www.w3.org/} for ontology modeling, is essentially a description logic (see, e.g., \citealp{fost}, for an introduction to OWL and an in-depth description of the correspondences). As such, DLs are currently gaining significant momentum in application areas, and are being picked up as knowledge representation paradigm by both industry and applied research.

The DL known as \SHIQ{} is among the most prominent DL fragments that do not feature nominals,\footnote{Nominals, i.e., concepts that denote a set with exactly one element, usually cause a reasoning efficiency problem when added to \SHIQ. This is evident from the performance of existing systems, and finds its theoretical justification in the fact that they increase worst-case complexity from ExpTime-completeness to NExpTime-completeness.} and it covers most of the OWL language. Various OWL reasoners implement efficient reasoning support for \SHIQ{} by means of tableau methods, e.g., Pellet,\footnote{http://clarkparsia.com/pellet/} FaCT++,\footnote{http://owl.man.ac.uk/factplusplus/} or RacerPro,\footnote{http://www.racer-systems.com/}.

However, even the most efficient implementations of reasoning algorithms to date do not scale up to very data-intensive application scenarios. This motivates the search for alternative reasoning approaches that build upon different methods in order to address cases where tableau algorithms turn out to have certain weaknesses. Successful examples are KAON2 \citep{kaon2} based on resolution, HermiT \citep{hermit} based on hyper-tableaux, as well as the consequence-based systems CB \citep{kazakov09-cdreasoning}, ConDOR \citep{conf/ijcai/Simancik11}, and ELK \citep{KKS11:parallEL}. Moreover, especially for lightweight DLs, approaches based on rewriting queries \citep{DBLP:journals/jar/CalvaneseGLLR07} or both queries and data \citep{DBLP:conf/kr/KontchakovLTWZ10} have been proposed.

In this paper, we propose the use of a variant of \textit{type elimination}, a notion first introduced by \cite{pratt}, as a reasoning paradigm for DLs. To implement the necessary computations on large type sets in a compressed way, we suggest the use of ordered binary decision diagrams (OBDDs). \mbox{OBDDs} have been applied successfully in the domain of large-scale model checking and verification, but have hitherto seen only little investigation in DLs, e.g., by \cite{pansattlervardi}.

Most of the description logics considered in this article exhibit restricted Boolean role expressions as a non-standard modeling feature, which is indicated by a $b$ or (if further restricted) $b_s$ in the name of the DL. In particular, we propose a novel method for reasoning in \SHIQb{} knowledge bases featuring terminological and assertional knowledge including (in)equality statements as well as DL-safe rules.

Our work starts by considering terminological reasoning in the DL \ALCIb{}, which is less expressive than \SHIQb{}. We introduce a method that compiles an \ALCIb{} terminology into an OBDD representation. Thereafter, we show that the output of this algorithm can be used for generating a disjunctive Datalog program that can in turn be combined with ABox data to obtain a correct reasoning procedure. Finally, the results for \ALCIb{} are lifted to full \SHIQb{} by providing an appropriate translation from the latter to the former.

This article combines and consolidates our previous work about pure TBox reasoning \citep{RKH-08-OBDDs}, its extension to ABoxes \citep{RKH:OBBD08b} and some notes on reasoning in DLs with Boolean role expressions \citep{RKH:Jelia-08} by
\begin{iteMize}{$\bullet$}
\item providing a collection of techniques for eliminating \SHIQb{} modeling
features that impede the use of our type elimination approach,
\item laying out the model-theoretic foundations for type-elimination-based
reasoning for very expressive description logics without nominals, using the \emph{domino} metaphor for 2-types,
\item elaborating the possibility of using OBDDs for making
type elimination computationally feasible,
\item providing a canonical translation of OBDDs into disjunctive
Datalog to enable reasoning with assertional information, and
\item making the full proofs accessible in a published version.
\end{iteMize}
Moreover, we extend our work by adding some missing aspects and completing the theoretical investigations by

\begin{iteMize}{$\bullet$}
\item extending the procedures for reducing \SHIQb{} to \ALCIb{} to ABoxes and DL-safe rules,
\item establishing worst-case optimality of our algorithms,
\item extending the supported language: while our previous work only covered
terminological reasoning in \SHIQ{} \citep{RKH-08-OBDDs} and combined reasoning in \ALCIb{} \citep{RKH:OBBD08b}, we now support reasoning in \SHIQb{} knowledge bases featuring terminological and assertional knowledge, including (in)equality statements and DL-safe rules.
\end{iteMize}

The structure of this article is as follows. Section~\ref{sec:prelims} recalls relevant preliminaries. Section~\ref{sec:dominoes} discusses the computation of sets of \emph{dominoes} that represent models of \ALCIb{} knowledge bases. Section~\ref{sec:boolfunc} casts this computation into a manipulation of OBDDs as underlying data structures. Section~\ref{sec:abox} discusses how the resulting OBDD presentation can be transformed to disjunctive Datalog and establishes the correctness of the approach. Section~\ref{sec:reduction} provides a transformation from \SHIQb{} to \ALCIb{}, thereby extending the applicability of the proposed method to \SHIQb{} knowledge bases. Section \ref{sec:related} discusses related work and Section \ref{sec:conc} concludes.

\section{The Description Logics \SHIQb{} and \ALCIb{}}\label{sec:prelims}

%
%
We first recall some basic definitions of DLs and introduce our notation. A more gentle first introduction to DLs, together with pointers to further reading, is given in \cite{DBLP:conf/rweb/Rudolph11}. Here, we define a rather expressive description logic \SHIQb{} that extends \SHIQ{} with restricted Boolean role expressions (see, e.g., \citealp{Tobies:PhD}).

\begin{definition}
A \SHIQb{} knowledge base is based on three disjoint sets of \define{concept names} $\connames$, \define{role names} $\rolnames$, and \define{individual names} $\indnames$. The set of \define{atomic roles} $\lang{R}$ is defined by $\lang{R}\defeq \rolnames\cup\{\rolR^- \mid \rolR\in\rolnames\}$. In addition, we let $\Inv(\rolR)\defeq\rolR^-$ and $\Inv(\rolR^-) \defeq \rolR$, and we extend this notation also to sets of atomic roles. In the following, we use the symbols $\rolR$ and $\rolS$ to denote atomic roles, if not specified otherwise.

The set of \define{Boolean role expressions} $\lang{B}$ is defined as
\[\lang{B} \Coloneqq \Rlang \mid \neg \lang{B} \mid \lang{B} \sqcap \lang{B} \mid \lang{B} \sqcup \lang{B}.\]
We use $\roleval{}{}$ to denote entailment between sets of atomic roles and role expressions. Formally, given a set $\mathscr{R}$ of atomic roles, we inductively define:
\begin{iteMize}{$\bullet$}
\item for atomic roles $\rolR$, $\roleval{\mathscr{R}}{\rolR}$ if $\rolR\in\mathscr{R}$, and $\nroleval{\mathscr{R}}{\rolR}$ otherwise,
\item $\roleval{\mathscr{R}}{\neg\rolU}$ if $\nroleval{\mathscr{R}}{\rolU}$, and $\nroleval{\mathscr{R}}{\neg\rolU}$ otherwise,
\item $\roleval{\mathscr{R}}{\rolU\sqcap\rolV}$ if $\roleval{\mathscr{R}}{\rolU}$ and $\roleval{\mathscr{R}}{\rolV}$, and $\nroleval{\mathscr{R}}{\rolU\sqcap\rolV}$ otherwise,
\item $\roleval{\mathscr{R}}{\rolU\sqcup\rolV}$ if $\roleval{\mathscr{R}}{\rolU}$ or $\roleval{\mathscr{R}}{\rolV}$, and $\nroleval{\mathscr{R}}{\rolU\sqcup\rolV}$ otherwise.
\end{iteMize}
A Boolean role expression $\rolU$ is \define{restricted} if $\nroleval{\emptyset}{\rolU}$. The set of all restricted role expressions is denoted by $\lang{T}$, and the symbols $\rolU$ and $\rolV$ will be used throughout this paper to denote restricted role expressions.
 A \SHIQb{} \define{RBox} is a set of axioms of the form $\rolU\ssb\rolV$ (role inclusion axiom) or $\Trans(\rolR)$ (transitivity axiom). The set of \define{non-simple} roles (for a given RBox) is defined as the smallest subset of $\Rlang$ satisfying:
\begin{iteMize}{$\bullet$}
\item If there is an axiom $\Trans(\rolR)$, then $\rolR$ is non-simple.
\item If there is an axiom $\rolR\ssb\rolS$ with $\rolR$ non-simple, then $\rolS$ is non-simple.
\item If $\rolR$ is non-simple, then $\Inv(\rolR)$ is non-simple.
\end{iteMize}
An atomic role is \define{simple} if it is not non-simple. In \SHIQb{}, every non-atomic Boolean role expression must contain only simple roles.

Based on a \SHIQb{} RBox, the set of \define{concept expressions} $\lang{C}$ is defined as
\[\lang{C} \Coloneqq \connames \mid \top \mid \bot \mid \neg \lang{C} \mid \lang{C} \sqcap \lang{C} \mid \lang{C} \sqcup \lang{C}\mid \forall\lang{T}.\lang{C}\mid \exists\lang{T}.\lang{C} \mid \atmost{n}\lang{R}.\lang{C} \mid \atleast{(n+1)}\lang{R}.\lang{C},\]
%
where $n\geq 0$ denotes a natural number, and the role $S$ in expressions $\atmost{n}S.C$ and $\atleast{(n+1)}S.C$ is required to be simple.
Common names for the various forms of concept expressions are given in Table~\ref{table:SHIQb} (lower part).
%
%
Throughout this paper, the symbols $\conC$, $\conD$ will be used to denote concept expressions. A \SHIQb{} \define{TBox} (or \define{terminology}) is a set of \define{general concept inclusion axioms} (GCIs) of the form $\conC \ssb \conD$.
%
\end{definition}

Besides the terminological components, DL knowledge bases typically include assertional knowledge as well. In order to increase expressivity and to allow for a uniform presentation of our approach we generalize this by allowing knowledge bases to contain so-called DL-safe rules as introduced by \cite{DLsafe-JWS}.
\newcommand{\Vlang}{\mathbf{V}}
\begin{definition}\label{defn:rules}
Let $\Vlang$ be a countable set of first-order variables. A \define{term} is an element of $\Vlang\cup\indnames$. Given terms $t$ and $u$, a \emph{concept atom/role atom/equality atom} is a formula of the form $\conC(t)/\rolR(t,u)/t\approx u$ with $\conC\in\connames$ and $\rolR\in\rolnames$. A \define{DL-safe rule} for \SHIQb{} is a formula $B\to H$, where $B$ and $H$ are possibly empty conjunctions of (role, concept, and equality) atoms. To simplify notation, we will often use finite sets $S$ of atoms for representing the conjunction $\bigwedge S$.

A set $\mathscr{P}$ of DL-safe rules is called a \define{rule base}. An \define{extended} \SHIQb{} \define{knowledge base} $\kb$ is a triple $\tuple{\mathscr{T},\mathscr{R},\mathscr{P}}$, where $\mathscr{T}$ is a \SHIQb{} TBox, $\mathscr{R}$ is a \SHIQb{} RBox, and $\mathscr{P}$ is a rule base.
\end{definition}

We only consider extended knowledge bases in this work, so we will often just speak of knowledge bases. In the literature, a DL ABox is usually allowed to contain assertions of the form $\conA(a)$, $\rolR(a,b)$, or $a \approx b$, where $a,b\in\indnames$, $\conA\in\connames$, and $\rolR\in\rolnames$. We assume that all roles and concepts occurring in the ABox are atomic.\footnote{This common assumption is made without loss of generality in terms of knowledge base expressivity. It is essential for defining the ABox-specific complexity measure of \emph{data complexity}, although it might be questionable in cases where ABox statements with complex concept expressions belong to the part of the knowledge base that is frequently changing.} These assertions can directly be expressed as DL-safe rules that have empty (vacuously true) bodies and a single head atom. Conversely, the negation of these assertions can be expressed by rules that have the assertion as body atom while having an empty (vacuously false) head. Knowing this, we will not specifically consider assertions or negated assertions in the proofs of this paper. For convenience we will, however, sometimes use the above notations instead of their rule counterparts when referring to (positive or negated) ground facts.

As mentioned above, we will mostly consider fragments of \SHIQb{}. In particular, an (extended) \ALCIb{} knowledge base is an (extended) \SHIQb{} knowledge base that contains no RBox axioms and no number restrictions (i.e., concept expressions $\atmost{n}\rolR.\conC$ or $\atleast{n}\rolR.\conC$). Consequently, an extended \ALCIb{} knowledge base only consists of a pair $\tuple{\mathscr{T},\mathscr{P}}$, where $\mathscr{T}$ is a TBox and $\mathscr{P}$ is a rule base. The related DL \ALCIQb{} has been studied by \cite{Tobies:PhD}. \fancypicture{Figure~\ref{fig:DLs} give a brief exemplary overview about the essentials of \SHIQ{} and \ALCIb.
\begin{figure}[bht]
\includegraphics[width=1\textwidth]{figures/executiveDLs}
\caption{Modeling features in \SHIQ{} and \ALCIb{}, respectively.\label{fig:DLs}}
\end{figure}}

\begin{table}[t]
 {\normalsize
 \begin{tabular}{|l|l|l|}
  \hline
 Name & Syntax & Semantics \\  \hline
 & & \\[-2ex]
 inverse role & $\rolR^-$ & $\{\tuple{x,y}\in\Delta^\Inter\times\Delta^\Inter \mid \tuple{y,x} \in \rolR^\Inter\}$ \\
 role negation & $\neg \rolU$ & $\{\tuple{x,y}\in\Delta^\Inter\times\Delta^\Inter \mid \tuple{x,y} \not\in \rolU^\Inter\}$ \\
 role conjunction & $\rolU \sqcap \rolV$ & $\rolU^\Inter \cap \rolV^\Inter$ \\
 role disjunction & $\rolU \sqcup \rolV$ & $\rolU^\Inter \cup \rolV^\Inter$ \\\hline
 & & \\[-2ex]
 top & $\top$ & $\Delta^\Inter$\\  
 bottom & $\bot$ & $\emptyset$ \\  
 negation & $\neg \conC$& $\Delta^\Inter \setminus \conC^{\Inter}$\\  
 conjunction & $\conC\sqcap \conD$& $\conC^{\Inter}\cap \conD^{\Inter}$ \\  
 disjunction & $\conC\sqcup \conD$& $\conC^{\Inter}\cup \conD^{\Inter}$\\  
 universal restriction & $\forall \rolU.\conC$ & $\{x\in\Delta^\Inter \mid \tuple{x,y} \in \rolU^{\Inter} \text{ implies } y\in \conC^{\Inter}\}$\\  
 existential restriction & $\exists \rolU.\conC$ & $\{x\in\Delta^\Inter \mid \tuple{x,y}\in\rolU^{\Inter}$, $y\in \conC^{\Inter}$ for some $y\in\Delta^\Inter\}$\\  
 qualified & $\atmost{n}\rolS.C$ & $\{x\in\Delta^\Inter \mid \#\{y\,\mathord{\in}\,\Delta^\Inter|\tuple{x,y}\,\mathord{\in}\,\rolS^{\Inter}\text{, } y\,\mathord{\in}\,\conC^{\Inter}\}\le n\}$\\
 \phantom{xii}number restriction & $\atleast{n}\rolS.C$ & $\{x\in\Delta^\Inter \mid \#\{y\,\mathord{\in}\,\Delta^\Inter|\tuple{x,y}\,\mathord{\in}\,\rolS^{\Inter}\text{, } y\,\mathord{\in}\,\conC^{\Inter}\}\ge n\}$\\  \hline
 \end{tabular}}
\caption{Semantics of constructors in \SHIQb{} for an interpretation $\Inter$ with domain $\Delta^\Inter$\label{table:SHIQb}}
\end{table}

The semantics of \SHIQb{} and its sublogics is defined in the usual, model-theoretic way. An interpretation $\Inter$ consists of a set $\Delta^\Inter$ called \emph{domain} (the elements of it being called \define{individuals}) together with a function $\cdot^\Inter$ mapping individual names to elements of $\Delta^\Inter$, concept names to subsets of $\Delta^\Inter$, and role names to subsets of $\Delta^\Inter\times\Delta^\Inter$.

The function $\cdot^\Inter$ is extended to role and concept expressions as shown in Table~\ref{table:SHIQb}. An interpretation $\Inter$ \define{satisfies} an axiom $\varphi$ if we find that $\Inter\models\varphi$, where
\begin{iteMize}{$\bullet$}
\item $\Inter\models \rolU\ssb\rolV$ if $\rolU^\Inter \subseteq \rolV^\Inter$,
\item $\Inter\models \Trans(\rolR)$ if $\rolR^\Inter$ is a transitive relation,
\item $\Inter\models \conC\ssb\conD$ if $\conC^\Inter \subseteq \conD^\Inter$,
\end{iteMize}
$\Inter$ \define{satisfies} a knowledge base $\kb$, denoted $\Inter\models\kb$, if it satisfies all axioms of $\kb$.

It remains to define the semantics of DL-safe rules. A (DL-safe) \define{variable assignment} $Z$ for an interpretation $\Inter$ is a mapping from the set of variables $\Vlang$ to $\{a^\Inter\mid a\in \indnames\}$. Given a term $t\in\indnames\cup\Vlang$, we set $t^{\Inter,Z}\defeq Z(t)$ if $t\in\Vlang$, and $t^{\Inter,Z}\defeq t^{\Inter}$ otherwise. Given a concept atom $\conC(t)$\  /\  role atom $\rolR(t,u)$\  /\  equality atom $t \approx u$, we write $\Inter,Z\models \conC(t)$ \ /\ $\Inter,Z\models \rolR(t,u)$ \ /\  $\Inter,Z\models t\approx u$ if $t^{\Inter,Z}\in\conC^{\Inter}$\  /\  $\tuple{t^{\Inter,Z},u^{\Inter,Z}}\in\rolR^{\Inter}$\  /\  $t^{\Inter,Z}=u^{\Inter,Z}$, and we say that $\Inter$ and $Z$ \define{satisfy} the atom in this case.

An interpretation $\Inter$ \emph{satisfies} a rule $B\to H$ if, for all variable assignments $Z$ for $\Inter$, either $\Inter$ and $Z$ satisfy all atoms in $H$, or $\Inter$ and $Z$ fail to satisfy some atom in $B$. In this case, we write $\Inter\models B\to H$ and say that $\Inter$ is a \define{model} for $B\to H$. An interpretation satisfies a rule base $\mathscr{P}$ (i.e., it is a \define{model} for it) whenever it satisfies all rules in it.
An extended knowledge base $\kb=\tuple{\mathscr{T},\mathscr{R},\mathscr{P}}$ is \define{satisfiable} if it has an interpretation $\Inter$ that is a model for $\mathscr{T}$, $\mathscr{R}$, and $\mathscr{P}$, and it is \define{unsatisfiable} otherwise.
\define{Satisfiability}, \define{equivalence}, and \define{equisatisfiability} of (extended) knowledge bases are defined as usual.

For convenience of notation, we abbreviate TBox axioms of the form $\top\ssb\conC$ by writing just $\conC$. Statements such as $\Inter\models\conC$ and $\conC\in\kb$ are interpreted accordingly. Note that $\conC\ssb\conD$ can thus be written as $\neg\conC\sqcup\conD$.


We often need to access a particular set of quantified and atomic subformulae of a DL concept expression. These specific \emph{parts} are provided by the function  $\parts: \mathbf{C} \to 2^\mathbf{C}$:
\[
\parts(\conC) \defeq \left\{
\begin{array}{@{}l@{}l}
\parts(\conD)                  & \text{if } \conC = \neg\conD, \\
\parts(\conD)\cup\parts(\conE)\; & \text{if } \conC = \conD\sqcap\conE
                                 \text{ or } \conC = \conD\sqcup\conE, \\
\{\conC\}\cup\parts(\conD)\;     & \text{if } \conC = \quantor\rolU.\conD\text{ with }\quantor\in\{\exists,\forall,\atleast{n}\!,\atmost{n}\!\}, \\
\{\conC\}                      & \text{otherwise.}
\end{array}\right.
\]
We generalize $\parts$ to DL knowledge bases\label{page:parts} $\kb$ by defining $\parts(\kb)$ to be the union of the sets $\parts(\conC)$ for all TBox axioms $\conC$ in $\kb$, where we express TBox axioms as simple concept expressions as explained above.

Given an extended knowledge base $\kb$, we obtain its negation normal form $\NNF(\kb)$ by keeping all RBox statements and DL-safe rules untouched and converting every TBox concept $C$ into its negation normal form $\NNF(C)$ in the usual, recursively defined way:
\begin{center}
\begin{tabular}{@{}ll}
\begin{tabular}{l@ {$\;\;\defeq\;\;$}l}
$\NNF(\neg \top)$ & $\bot$ \\
$\NNF(\neg \bot)$ & $\top$ \\
$\NNF(\conC)$ & $\conC$ if $\conC\in \{A,\neg A,\top,\bot\}$ \\
$\NNF(\neg\neg \conC)$ & $\NNF(\conC)$ \\
$\NNF(\conC\sqcap \conD)$ & $\NNF(\conC) \sqcap \NNF(\conD)$ \\
$\NNF(\neg(\conC\sqcap \conD))$ & $\NNF(\neg \conC) \sqcup \NNF(\neg \conD)$ \\
$\NNF(\conC\sqcup \conD)$ & $\NNF(\conC) \sqcup \NNF(\conD)$ \\
$\NNF(\neg(\conC\sqcup \conD))$ & $\NNF(\neg \conC) \sqcap \NNF(\neg \conD)$ \\
\end{tabular}&
\begin{tabular}{l@ {$\;\;\defeq\;\;$}l}
$\NNF(\forall \rolU.\conC)$ & $\forall \rolU.\NNF(\conC)$ \\
$\NNF(\neg\forall \rolU.\conC)$ & $\exists \rolU.\NNF(\neg \conC)$ \\
$\NNF(\exists \rolU.\conC)$ & $\exists \rolU.\NNF(\conC)$ \\
$\NNF(\neg\exists \rolU.\conC)$ & $\forall \rolU.\NNF(\neg \conC)$ \\
$\NNF(\atmost{n}\rolR.\conC)$ & $\atmost{n}\rolR.\NNF(\conC)$ \\
$\NNF(\neg\,\atmost{n}\rolR.\conC)$ & $\atleast{(n+1)}\rolR.\NNF(\conC)$ \\
$\NNF(\atleast{n}\rolR.\conC)$ & $\atleast{n}\rolR.\NNF(\conC)$ \\
$\NNF(\neg\,\atleast{n}\rolR.\conC)$ & $\atmost{(n-1)} \rolR.\NNF(\conC)$ \\
\end{tabular}\\
\end{tabular}
\end{center}
It is well known that $\kb$ and $\NNF(\kb)$ are semantically equivalent.

In places, we will additionally require another well-known normalization step that simplifies the structure of $\kb$ by \define{flattening} it to a knowledge base $\FLAT(\kb)$. This is achieved by transforming $\kb$ into negation normal form and exhaustively applying the following transformation rules:
\begin{iteMize}{$\bullet$}
\item Select an outermost occurrence of $\quantor\rolU.\conD$ in $\kb$, such that $\quantor\in\{\exists,\forall,\atmost{n}\!,\atleast{n}\!\}$ and $\conD$ is a non-atomic concept.
\item Substitute this occurrence with $\quantor\rolU.\conF$ where $\conF$
is a fresh concept name (i.e., one not occurring in the knowledge base).
\item If $\quantor\in\{\exists,\forall,\atleast{n}\!\}$, add $\neg\conF\sqcup\conD$ to the knowledge base.
\item If $\quantor = \atmost{n}\!$ add $\NNF(\neg\conD)\sqcup\conF$ to the
knowledge base.
\end{iteMize}
Obviously, this procedure terminates, yielding a flat knowledge base $\FLAT(\kb)$ all TBox axioms of which are $\sqcap,\sqcup$-expressions over formulae of the form $\top$, $\bot$, $\conA$, $\neg\conA$, or $\quantor\rolU.\conA$ with $\conA$ an atomic concept name. Flattening is known to be a satisfiability-preserving transformation; we include the proof for the sake of self-containedness.

\begin{proposition}\label{prop:flateq}
For every \SHIQb{} knowledge base $\kb$, we find that $\kb$ and $\FLAT(\kb)$ are equisatisfiable.
\end{proposition}
\proof
We first prove inductively that every model of $\FLAT(\kb)$ is a model of $\kb$. Let $\kb'$ be an intermediate knowledge base and let $\kb''$ be the result of applying one single substitution step to $\kb'$ as described in the above procedure. We now show that any model $\Inter$ of $\kb''$ is a model of $\kb'$. Let $\quantor\rolU.\conD$ be the concept expression substituted in $\kb'$. Note that after every substitution step, the knowledge base is still in negation normal form. Thus, we see that $\quantor\rolU.\conD$ occurs outside the scope of any negation or quantifier in a $\kb'$ axiom $\conE'$, and the same is the case for $\quantor\rolU.\conF$ in the respective $\kb''$ axiom $\conE''$ obtained after the substitution. Hence, if we show that $(\quantor\rolU.\conF)^\Inter \subseteq (\quantor\rolU.\conD)^\Inter$, we can conclude that $\conE''^\Inter \subseteq \conE'^\Inter$. From $\Inter$ being a model of $\kb''$ and therefore $\conE''^\Inter = \Delta^\Inter$, we would then easily derive that $\conE'^\Inter=\Delta^\Inter$ and hence find that $\Inter \models \kb'$, as all other axioms from $\kb'$ are trivially satisfied due to their presence in $\kb''$.

It remains to show $(\quantor\rolU.\conF)^\Inter \subseteq (\quantor\rolU.\conD)^\Inter$. To show this, consider some arbitrary $\delem \in (\quantor\rolU.\conF)^\Inter$. We distinguish various cases:
\begin{iteMize}{$\bullet$}
\item $\quantor = \atleast{n}$\\
Then there are distinct individuals $\delem_1,\ldots,\delem_n\in\Delta^\Inter$ with $\tuple{\delem,\delem_i} \in \rolU^\Inter$ and $\delem_i\in \conF^\Inter$ for $1\leq i\leq n$. Since $\neg\conF\sqcup\conD\in\kb''$, we have $\Inter\models\neg\conF\sqcup\conD$, and therefore $\delem_i\in \conD^\Inter$ for all the $n$ distinct $\delem_i$. Thus $\delem \in (\atleast{n}\rolU.\conF)^\Inter$.
\item $\quantor = \atmost{n}$\\
Then the number of individuals $\delem'\in\Delta^\Inter$ with $\tuple{\delem,\delem'} \in \rolU^\Inter$ and $\delem'\in \conF^\Inter$ is not greater than $n$. Since $\NNF(\neg\conD)\sqcup\conF\in\kb''$, we know $\conD^\Inter \subseteq \conF^\Inter$. Thus, also the number of individuals $\delem'\in\Delta^\Inter$ with $\tuple{\delem,\delem'} \in \rolU^\Inter$ and $\delem'\in \conD^\Inter$ cannot be greater than $n$, leading to the conclusion $\delem \in (\atmost{n}\rolU.\conD)^\Inter$. Hence, we have $(\atmost{n}\rolU.\conF)^\Inter \subseteq (\atmost{n}\rolU.\conD)^\Inter$.
\end{iteMize}
The arguments for $\quantor = \exists$ and $\quantor = \forall$ are very similar, since these cases can be treated like $\atleast{1}\rolU.\conF$ and $\atmost{0}\rolU.\neg\conF$, respectively. Thus we obtain $\delem \in (\quantor\rolU.\conD)^\Inter$ in each case as required.

For the other direction of the claim, note that every model $\Inter$ of $\kb$ can be transformed into a model $\Jnter$ of $\FLAT(\kb)$ by following the flattening process described above: Let $\kb''$ result from $\kb'$ by substituting $\quantor \rolU .\conD$ by $\quantor \rolU .\conF$ and adding the respective axiom. Furthermore, let $\Inter'$ be a model of $\kb'$. Now we construct the interpretation $\Inter''$ as follows: $\conF^{\Inter''} \defeq (\quantor \rolU .\conD)^{\Inter'}$ and for all other concept and role names $\con{N}$ we set $\con{N}^{\Inter''} \defeq \con{N}^{\Inter'}$. Then $\Inter''$ is a model of $\kb''$.
\qed

\section{Building Models from Domino Sets}\label{sec:dominoes}

In this section, we introduce the notion of a set of \define{dominoes} for a given \ALCIb{} TBox. Rules (and thus ABox axioms) will be incorporated in Section~\ref{sec:abox} later on. Intuitively, a domino abstractly represents two individuals in an \ALCIb{} interpretation, reflecting their satisfied concepts and mutual role relationships. Thereby, dominoes are conceptually very similar to the concept of 2-types, as used in investigations on two-variable fragments of first-order logic, e.g., by \cite{GradelOR97}. We will see that suitable sets of such two-element pieces suffice to reconstruct models of \ALCIb{}, which also reveals certain model-theoretic properties of this not so common DL. In particular, every satisfiable \ALCIb{} TBox admits tree-shaped models. This result is rather a by-product of our main goal of decomposing models into unstructured sets of local domino components, but it explains why our below constructions have some similarity with common approaches of showing tree-model properties by unraveling models.

After introducing the basics of our domino representation, we present an algorithm for deciding satisfiability of an \ALCIb{} terminology based on sets of dominoes.

\subsection{From Interpretations to Dominoes}\label{ssec:domint}

We now introduce the basic notion of a domino set, and its relationship to interpretations. Given a DL with concepts $\lang{C}$ and roles $\lang{R}$, a \define{domino} over $\mathscr{C}\subseteq\lang{C}$ is an arbitrary triple $\tuple{\mathscr{A},\mathscr{R},\mathscr{B}}$, where $\mathscr{A}, \mathscr{B}\subseteq \mathscr{C}$ and $\mathscr{R}\subseteq \lang{R}$. In the following, we will always assume a fixed language and refer to dominoes over that language only. We now formalize the idea of deconstructing an interpretation into a set of dominoes.

\begin{definition}\label{def:domp}
Given an interpretation $\Inter = \tuple{\Delta^\Inter, \cdot^\Inter}$, and a set $\mathscr{C}\subseteq\lang{C}$ of concept expressions, the \define{domino projection} of $\Inter$ w.r.t.~$\mathscr{C}$, denoted by $\domp_{\mathscr{C}}(\Inter)$ is the set that contains, for all $\delem,\delem'\in\Delta^\Inter$, the triple $\tuple{\mathscr{A},\mathscr{R},\mathscr{B}}$ with \vspace{-0.5ex}
\[
 \mathscr{A} = \{ \conC\in\mathscr{C} \mid \delem \in \conC^\Inter \},\qquad
 \mathscr{R} = \{ \rolR\in\lang{R} \mid \tuple{\delem,\delem'} \in \rolR^\Inter \},\qquad
 \mathscr{B} = \{ \conC\in\mathscr{C} \mid \delem' \in \conC^\Inter \}.
\]
\end{definition}

\fancypicture{Figure~\ref{fig:DominoCrea} gives an intuition how domino projections are defined and justifies the used metaphor.
\begin{figure}[bht]
\includegraphics[width=1\textwidth]{figures/executiveDomCreation}
\caption{The idea behind domino projections.\label{fig:DominoCrea}}
\end{figure}}
It is easy to see that domino projections do not faithfully represent the structure of the interpretation that they were constructed from. But, as we will see below, domino projections capture enough information to reconstruct models of a TBox ${\mathscr{T}}$, as long as $\mathscr{C}$ is chosen to contain at least $\parts({\mathscr{T}})$. For this purpose, we introduce the inverse construction of interpretations from arbitrary domino sets.

\begin{definition}\label{def:domint}
Given a set $\DO$ of dominoes, the induced \define{domino interpretation} $\Inter(\DO)=\tuple{\Delta^\Inter,\cdot^\Inter}$ is defined as follows:

\begin{enumerate}[(1)]
\item $\Delta^\Inter$ consists of all nonempty finite words over $\DO$ where, for each pair of subsequent letters $\tuple{\mathscr{A},\mathscr{R},\mathscr{B}}$ and $\tuple{\mathscr{A}',\mathscr{R}',\mathscr{B}'}$ in a word, we have
$\mathscr{B}=\mathscr{A}'$.
\item For a word $\selem = \tuple{\mathscr{A}_1,\mathscr{R}_1,\mathscr{A}_2} \tuple{\mathscr{A}_2,\mathscr{R}_2,\mathscr{A}_3} \ldots \tuple{\mathscr{A}_{i-1},\mathscr{R}_{i-1},\mathscr{A}_i}$ and a concept name $\conA\in\connames$, we define $\val(\selem)\defeq \mathscr{A}_i$ and set $\selem \in \conA^\Inter$ iff $\conA\in \val(\selem)$.
\item For a role name $\rolR\in\rolnames$, we set $\tuple{\selem_1,\selem_2}\in\rolR^\Inter$ if \vspace{-0.5ex}
\[\selem_2 = \selem_1 \tuple{\mathscr{A},\mathscr{R},\mathscr{B}} \text{ with }\rolR \in
\mathscr{R}  \quad\text{ or }\quad \selem_1 = \selem_2 \tuple{\mathscr{A},\mathscr{R},\mathscr{B}} \text{ with } \Inv(\rolR) \in \mathscr{R}.\]
\end{enumerate}
\end{definition}

We can now show that certain domino projections contain enough information to reconstruct models of a TBox.

\begin{proposition}\label{prop:idompi}
Consider a set $\mathscr{C}\subseteq\lang{C}$ of concept expressions, and an interpretation $\Jnter$, and let $\Knter\defeq \Inter(\domp_{\mathscr{C}}(\Jnter))$ denote the induced domino interpretation of the domino projection of $\Jnter$ w.r.t.\ $\mathscr{C}$. Then, for any \ALCIb{} concept expression $\conC\in\lang{C}$ with $\parts(\conC)\subseteq\mathscr{C}$, we have that  $\Jnter\models\conC$ iff $\Knter\models\conC$.

Especially, for any \ALCIb{} TBox ${\mathscr{T}}$, we have $\Jnter\models{\mathscr{T}}$ iff $\Inter(\domp_{\parts({\mathscr{T}})}(\Jnter))\models{\mathscr{T}}$.
\end{proposition}
\proof
Consider some $\conC\in\lang{C}$ as in the claim.
We first show the following: given any $\Jnter$-individual $\delem$ and $\Knter$-individual $\selem$ such that $\val(\selem) = \{\conD\in\mathscr{C}\mid \delem\in\conD^{\Jnter}\}$, we find that $\selem\in\conC^{\Knter}$ iff $\delem\in\conC^{\Jnter}$. Clearly, the overall claim follows from that statement using the observation that a suitable $\delem\in\Delta^{\Jnter}$ must exist for all $\selem\in\Delta^{\Knter}$ and vice versa. We proceed by induction over the structure of $\conC$, noting that $\parts(\conC)\subseteq\mathscr{C}$ implies $\parts(\conD)\subseteq\mathscr{C}$ for any subconcept $\conD$ of $\conC$.

The base case $\conC\in\connames$ is immediately satisfied by our assumption on the relationship of $\delem$ and $\selem$, since $\conC\in\parts(\conC)$. For the induction step, we first note that the case $\conC\in\{\top,\bot\}$ is also trivial. For $\conC = \neg\conD$ and $\conC = \conD\sqcap\conD'$ as well as $\conC = \conD\sqcup\conD'$, the claim follows immediately from the induction hypothesis for $\conD$ and $\conD'$.

Next consider the case $\conC=\exists\rolU.\conD$, and assume that $\delem\in\conC^{\Jnter}$. Hence there is some $\delem'\in\Delta^{\Jnter}$ such that $\tuple{\delem,\delem'}\in\rolU^{\Jnter}$ and $\delem'\in\conD^{\Jnter}$. Then the pair $\tuple{\delem,\delem'}$ generates a domino $\tuple{\mathscr{A},\mathscr{R},\mathscr{B}}$ and $\Delta^{\Knter}$ contains $\selem' = \selem\tuple{\mathscr{A},\mathscr{R},\mathscr{B}}$. $\tuple{\delem,\delem'}\in\rolU^{\Jnter}$ implies $\roleval{\mathscr{R}}{\rolU}$ (by definition of $\vdash$ and due to the fact that $\mathscr{R}$ contains exactly those $\rolR\in \lang{R}$ with $\tuple{\delem,\delem'}\in \rolR^\Jnter$), and hence $\tuple{\selem,\selem'}\in\rolU^{\Knter}$. Applying the induction hypothesis to $\conD$, we conclude $\selem'\in\conD^{\Knter}$. Now $\selem\in\conC^{\Knter}$ follows from the construction of $\Knter$.

For the converse, assume that $\selem\in\conC^{\Knter}$. Hence there is some $\selem'\in\Delta^{\Knter}$ such that $\tuple{\selem,\selem'}\in\rolU^{\Knter}$ and $\selem'\in\conD^{\Knter}$. By the definition of $\Knter$, there are two possible cases:
\begin{iteMize}{$\bullet$}
\item $\selem' = \selem\tuple{\val(\selem),\mathscr{R},\val(\selem')}$ and $\roleval{\mathscr{R}}{\rolU}$: Consider the two $\Jnter$-individuals $\tuple{\delem',\delem''}$ generating the domino $\tuple{\val(\selem),\mathscr{R},\val(\selem')}$. From $\selem'\in\conD^{\Knter}$ and the induction hypothesis, we obtain $\delem''\in\conD^\Jnter$. Together with $\tuple{\delem',\delem''}\in\rolU^{\Jnter}$ this implies $\delem'\in \conC^{\Jnter}$. Since $\conC = \exists\rolU.\conD \in \mathscr{C}$, we also have $\conC \in \val(\selem)$ and thus $\delem\in\conC^{\Jnter}$ as claimed.
\item $\selem = \selem'\tuple{\val(\selem'),\mathscr{R},\val(\selem)}$ and $\roleval{\Inv(\mathscr{R})}{\rolU}$: This case is similar to the first case, merely exchanging the order of $\tuple{\delem',\delem''}$ and using $\Inv(\mathscr{R})$ instead of $\mathscr{R}$.
\end{iteMize}

Finally, the case $\conC=\forall\rolU.\conD$ is dual to the case $\conC=\exists\rolU.\conD$, and we will omit the repeated argument. Note, however, that this case does not follow from the semantic equivalence of $\forall\rolU.\conD$ and $\neg\exists\rolU.\neg\conD$, since the proof hinges upon the fact that $\neg\conD$ is contained in $\mathscr{C}$ which is not given directly.
\qed

\fancypicture{Figure~\ref{fig:DominoRec} provides an intuition about the correspondence just proven: given a model, its domino projection can be used to ``reconstruct'' a model in a generic way. In the general case, however, the constructed model is structurally different from the original one. In particular, it is tree-like.
\begin{figure}[bht]
\begin{center}
\includegraphics[width=0.8\textwidth]{figures/executiveModelReconstructionCut}
\caption{Reconstructing models from domino projections.\label{fig:DominoRec}}
\end{center}
\end{figure}}

\subsection{Constructing Domino Sets}

As shown in the previous section, the domino projection of a model of an \ALCIb{} TBox can contain enough information for reconstructing a model. This observation can be the basis for designing an algorithm that decides TBox satisfiability. Usually (especially in tableau-based algorithms), checking satisfiability amounts to the attempt to construct a (representation of a) model. As we have seen, in our case it suffices to try to construct just a model's domino projection. If this can be done, we know that there is a model, if not, there is none.

In what follows, we first describe the iterative construction of such a domino set from a given TBox, and then show that it is indeed a decision procedure for TBox satisfiability.

\begin{algorithm}[t]
  \begin{algorithmic} [1]
    \REQUIRE \hspace*{1.5ex} ${\mathscr{T}}$ an \ALCIb{} TBox, $\mathscr{C}=\parts(\FLAT({\mathscr{T}}))$
    \ENSURE the canonical domino set $\DO_{\mathscr{T}}$ of ${\mathscr{T}}$\\
    \STATE initialize $\DO_0$ as the set of all dominoes $\tuple{\mathscr{A},\mathscr{R},\mathscr{B}}$ over $\mathscr{C}$ satisfying:
\STATE \hspace{7mm} for all $\conC\in\FLAT({\mathscr{T}})$, the GCI $\bigsqcap_{\conD\in\mathscr{A}}\conD\sqcap \bigsqcap_{\conD\in\mathscr{C}\setminus\mathscr{A}}\neg\conD\ssb\conC$ is a tautology\footnote{Please note that the formulae in $\FLAT({\mathscr{T}})$ and in $\mathscr{A}\subseteq\mathscr{C}$ are such that this can easily be checked by evaluating the Boolean operators in $\conC$ as if $\mathscr{A}$ was a set of true propositional variables.}\hfill\textbf{(kb)}
\STATE \hspace{7mm} for all $\exists\rolU.\conA \in \mathscr{C}$ with $\conA\in\mathscr{B}$ and $\roleval{\mathscr{R}}{\rolU}$, we have $\exists\rolU.\conA \in \mathscr{A}$,\hfill\textbf{(ex)}
\STATE \hspace{7mm} for all $\forall\rolU.\conA \in \mathscr{C}$ with $\forall\rolU.\conA \in \mathscr{A}$ and $\roleval{\mathscr{R}}{\rolU}$, we have $\conA\in\mathscr{B}$.\hfill\textbf{(uni)}
    \STATE i := 0
\REPEAT
 \STATE i := i+1
 \STATE determine $\DO_{i}$ as the set of all dominoes $\tuple{\mathscr{A},\mathscr{R},\mathscr{B}}\in\DO_{i-1}$ satisfying:
 \STATE\hspace{3mm} for all $\exists\rolU.\conA \in \mathscr{A}$, there is some $\tuple{\mathscr{A},\mathscr{R}',\mathscr{B}'}\in\DO_{i-1}$ with $\roleval{\mathscr{R}'}{\rolU}$ and $\conA\in\mathscr{B}'$,\hfill\textbf{(delex)}
 \STATE\hspace{3mm} for all $\forall\rolU.\conA \in \mathscr{C} \setminus \mathscr{A}$, there is some $\tuple{\mathscr{A},\mathscr{R}',\mathscr{B}'}\in\DO_{i-1}$ with $\roleval{\mathscr{R}'}{\rolU}$ but $\conA\notin\mathscr{B}'$,\hfill\textbf{(deluni)}
 \STATE\hspace{3mm} $\tuple{\mathscr{B},\Inv(\mathscr{R}),\mathscr{A}}\in\DO_{i-1}$.\hfill\textbf{(sym)}~\mbox{}%
\UNTIL $\DO_{i} = \DO_{i-1}$
\STATE $\DO_{\mathscr{T}}:=\DO_{i}$
\RETURN $\DO_{\mathscr{T}}$
  \end{algorithmic}
  \caption{Computing the canonical domino set $\DO_{\mathscr{T}}$ of a TBox ${\mathscr{T}}$}
  \label{def:builddoms}
\end{algorithm}

Algorithm~\ref{def:builddoms} describes the construction of the canonical domino set $\DO_{\mathscr{T}}$ of an \ALCIb{} TBox $\mathscr{T}$.\footnotetext[7]{Please note that the formulae in $\FLAT({\mathscr{T}})$ and in $\mathscr{A}\subseteq\mathscr{C}$ are such that this can easily be checked by evaluating the Boolean operators in $\conC$ as if $\mathscr{A}$ was a set of true propositional variables.}
Thereby, roughly speaking, condition \textbf{kb} ensures that all the concept parts $\mathscr{A}$ and $\mathscr{B}$ of the constructed domino set abide by the axioms of the considered TBox. The condition \textbf{ex} guarantees that, in every domino $\tuple{\mathscr{A},\mathscr{R},\mathscr{B}}$, the concept set $\mathscr{A}$ must contain all the existential concepts for which $\mathscr{R}$ and $\mathscr{B}$ serve as witnesses. Conversely, \textbf{uni} makes sure that every universally quantified concept recorded in $\mathscr{A}$ is appropriately propagated to $\mathscr{B}$, given a suitable $\mathscr{R}$. Once enforced, the conditions \textbf{kb}, \textbf{ex}, and \textbf{uni} remain valid even if the domino set is reduced further, hence they need to be taken care of only at the beginning of the algorithm. In contrast, the conditions \textbf{delex}, \textbf{deluni}, and \textbf{sym} may be invalidated again by removing dominoes from the set, thus they need to be applied in an iterated way until a fixpoint is reached. Condition \textbf{delex} removes all dominoes with the concept set $\mathscr{A}$ if $\mathscr{A}$ contains an existential concept for which no appropriate ``witness'' domino (in the above sense) can be found in the set. Likewise, \textbf{deluni} removes all dominoes with the concept set $\mathscr{A}$ if $\mathscr{A}$ does \emph{not} contain a universal concept which should hold given all the remaining dominoes. Finally, \textbf{sym} ensures that the domino set contains only dominoes that do have a ``symmetric partner'', i.e., one that is created by swapping $\mathscr{A}$ with $\mathscr{B}$ and inverting all of $\mathscr{R}$.

Given that every domino $\langle\mathscr{A},\mathscr{R},\mathscr{B}\rangle$ satisfies $\mathscr{A},\mathscr{B}\subseteq \mathscr{C}$ and $\mathscr{R} \subseteq \mathbf{R}$, and that both $\mathscr{C}$ and $\mathbf{R}$ are linearly bounded by the size of $\mathscr{T}$, $\DO_0$ is exponential in the size of the TBox, hence the iterative deletion of dominoes must terminate after at most exponentially many steps. Below we will show that this procedure is indeed sound and complete for checking TBox satisfiability. Before that, we will show a canonicity result for $\DO_{\mathscr{T}}$.

\begin{lemma}\label{lemma:domgreatest}
Consider an \ALCIb{} terminology $\mathscr{T}$ and an arbitrary model $\Inter$ of $\mathscr{T}$. Then the domino projection $\domp_{\parts(\FLAT(\mathscr{T}))}(\Inter)$ is contained in $\DO_{\mathscr{T}}$.
\end{lemma}
\begin{proof}
The claim is shown by a simple induction over the construction of $\DO_{\mathscr{T}}$. In the following, we use $\tuple{\mathscr{A},\mathscr{R},\mathscr{B}}$ to denote an arbitrary domino of $\domp_{\parts(\FLAT(\mathscr{T}))}(\Inter)$.
For the base case, we must show that $\domp_{\parts(\FLAT(\mathscr{T}))}(\Inter)\subseteq\DO_0$. Let $\tuple{\mathscr{A},\mathscr{R},\mathscr{B}}$ to denote an arbitrary domino of $\domp_{\parts(\FLAT(\mathscr{T}))}(\Inter)$ which was generated from elements $\tuple{\delem,\delem'}$. Then $\tuple{\mathscr{A},\mathscr{R},\mathscr{B}}$ satisfies condition \textbf{kb}, since $\delem\in\conC^{\Inter}$ for any $\conC\in\FLAT(\mathscr{T})$. 
The conditions \textbf{ex} and \textbf{uni} are obviously satisfied.

For the induction step, assume that $\domp_{\parts(\FLAT(\mathscr{T}))}(\Inter)\subseteq\DO_i$, and let $\tuple{\mathscr{A},\mathscr{R},\mathscr{B}}$ again denote an arbitrary domino of $\domp_{\parts(\FLAT(\mathscr{T}))}(\Inter)$ which was generated from elements $\tuple{\delem,\delem'}$.

\begin{iteMize}{$\bullet$}
\item For \textbf{delex}, note that $\exists\rolU.\conA\in\mathscr{A}$ implies $\delem\in(\exists\rolU.\conA)^{\Inter}$. Thus there is an individual $\delem''$ such that $\tuple{\delem,\delem''}\in\rolU^{\Inter}$ and $\delem''\in\conA^{\Inter}$. Clearly, the domino generated by $\tuple{\delem,\delem''}$ satisfies the conditions of \textbf{delex}.
\item For \textbf{deluni}, note that $\forall\rolU.\conA\not\in\mathscr{A}$ implies $\delem\notin(\forall\rolU.\conA)^{\Inter}$. Thus there is an individual $\delem''$ such that $\tuple{\delem,\delem''}\in\rolU^{\Inter}$ and $\delem''\notin\conA^{\Inter}$. Clearly, the domino generated by $\tuple{\delem,\delem''}$ satisfies the conditions of \textbf{deluni}.
\item The condition of \textbf{sym} for $\tuple{\mathscr{A},\mathscr{R},\mathscr{B}}$ is clearly satisfied by the domino generated from $\tuple{\delem',\delem}$.
\rightqed
\end{iteMize}
Therefore, the considered domino $\tuple{\mathscr{A},\mathscr{R},\mathscr{B}}$ must be contained in $\DO_{i+1}$ as well.
\end{proof}

Note that, in contrast to tableau procedures, the presented algorithm starts with a large set of dominoes and successively deletes undesired dominoes. Indeed, we will soon show that the constructed domino set is the largest such set from which a domino model can be obtained. The algorithm thus may seem to be of little practical use. In Section~\ref{sec:boolfunc}, we therefore refine the above algorithm to employ Boolean functions as implicit representations of domino sets, such that the efficient computational methods of OBDDs can be exploited. In the meantime, however, domino sets will serve us well for showing the required correctness properties.

An important property of domino interpretations constructed from canonical domino sets is that the (semantic) concept membership of an individual can typically be (syntactically) read from the domino it has been constructed of.

\begin{lemma}\label{lemma:domsound}
Consider an \ALCIb{} TBox ${\mathscr{T}}$ with nonempty canonical domino set $\DO_{\mathscr{T}}$, and define  $\mathscr{C}\defeq \parts(\FLAT({\mathscr{T}}))$ and $\Inter=\tuple{\Delta^\Inter,\cdot^{\Inter}}\defeq\Inter(\DO_{\mathscr{T}})$. Then, for all $\conC\in\mathscr{C}$ and $\selem\in\Delta^\Inter$, we have that $\selem\in\conC^\Inter$ iff $\conC\in\val(\selem)$. Moreover, $\Inter\models\FLAT({\mathscr{T}})$.
\end{lemma}
\proof
First note that the domain of $\Inter$ is nonempty whenever $\DO_{\mathscr{T}}$ is. Now if $\conC\in\connames$ is an atomic concept, the first claim follows directly from the definition of $\Inter$. The remaining cases that may occur in $\parts(\FLAT({\mathscr{T}}))$ are $\conC=\exists\rolU.\conA$ and $\conC=\forall\rolU.\conA$.

First consider the case $\conC=\exists\rolU.\conA$, and assume that $\selem\in\conC^\Inter$. Thus there is $\selem'\in\Delta^\Inter$ with $\tuple{\selem,\selem'}\in\rolU^{\Inter}$ and $\selem' \in\conA^\Inter$. The construction of the domino model admits two possible cases:
\begin{iteMize}{$\bullet$}
\item $\selem'=\selem\tuple{\val(\selem),\mathscr{R},\val(\selem')}$ with
$\roleval{\mathscr{R}}{\rolU}$ and $\conA\in\val(\selem')$. Since $\DO_{\mathscr{T}}\subseteq\DO_0$, we find that $\tuple{\val(\selem),\mathscr{R},\val(\selem')}$ satisfies condition \textbf{ex}, and thus $\conC\in\val(\selem)$ as required.
\item $\selem=\selem'\tuple{\val(\selem'),\mathscr{R},\val(\selem)}$ with $\roleval{\Inv(\mathscr{R})}{\rolU}$ and $\conA\in\val(\selem')$. By condition \textbf{sym}, $\DO_{\mathscr{T}}$ also contains the domino $\tuple{\val(\selem),\Inv(\mathscr{R}),\val(\selem')}$, and we can again invoke \textbf{ex} to conclude $\conC\in\val(\selem)$.
\end{iteMize}
For the other direction, assume $\exists\rolU.\conA \in \val(\selem)$. Thus $\DO_{\mathscr{T}}$ must contain some domino $\tuple{\mathscr{A},\mathscr{R},\val(\selem)}$, and by \textbf{sym} also the domino $\tuple{\val(\selem),\Inv(\mathscr{R}),\mathscr{A}}$. By condition \textbf{delex}, the latter implies that $\DO_{\mathscr{T}}$ contains a domino $\tuple{\val(\selem),\mathscr{R}',\mathscr{A}'}$. According to \textbf{delex}, we find that $\selem' = \selem\tuple{\val(\selem),\mathscr{R}',\mathscr{A}'}$ is an $\Inter$-individual such that $\tuple{\selem,\selem'}\in\rolU^{\Inter}$ and $\selem'\in\conA^\Inter$. Thus $\selem\in (\exists\rolU.\conA)^{\Inter}$ as claimed.

For the second case, consider $\conC=\forall\rolU.\conA$ and assume that $\selem\in\conC^\Inter$. Then $\DO_{\mathscr{T}}$ contains some domino $\tuple{\mathscr{A},\mathscr{R},\val(\selem)}$, and by \textbf{sym} also the domino $\tuple{\val(\selem),\Inv(\mathscr{R}),\mathscr{A}}$. For a contradiction, suppose that $\forall\rolU.\conA\not\in\val(\selem)$. By condition \textbf{deluni}, the latter implies that $\DO_{\mathscr{T}}$ contains a domino $\tuple{\val(\selem),\mathscr{R}',\mathscr{A}'}$. According to \textbf{deluni}, we find that $\selem' = \selem\tuple{\val(\selem),\mathscr{R}',\mathscr{A}'}$ is an $\Inter$-individual such that $\tuple{\selem,\selem'}\in\rolU^{\Inter}$ and $\selem'\notin\conD^\Inter$. But then $\selem\notin (\forall\rolU.\conA)^{\Inter}$, yielding the required contradiction.

For the other direction, assume that $\forall\rolU.\conA \in \val(\selem)$. According to the construction of the domino model, there are two possible cases for elements $\selem'$ with $\tuple{\selem,\selem'}\in\rolU^{\Inter}$:
\begin{iteMize}{$\bullet$}
\item $\selem'=\selem\tuple{\val(\selem),\mathscr{R},\val(\selem')}$ with
$\roleval{\mathscr{R}}{\rolU}$. Since $\DO_{\mathscr{T}}\subseteq\DO_0$, $\tuple{\val(\selem),\mathscr{R},\val(\selem')}$ must satisfy condition \textbf{uni}, and thus $\conA\in\val(\selem')$.
\item $\selem=\selem'\tuple{\val(\selem'),\mathscr{R},\val(\selem)}$ with $\roleval{\Inv(\mathscr{R})}{\rolU}$. By condition \textbf{sym}, $\DO_{\mathscr{T}}$ also contains the domino $\tuple{\val(\selem),\Inv(\mathscr{R}),\val(\selem')}$, and we can again invoke \textbf{uni} to conclude $\conA\in\val(\selem')$.
\end{iteMize}
Thus, $\conA\in\val(\selem')$ for all $\rolU$-successors $\selem'$ of $\selem$, and hence $\selem\in(\forall\rolU.\conA)^{\Inter}$ as claimed.\medskip

For the rest of the claim, note that any domino $\tuple{\mathscr{A},\mathscr{R},\mathscr{B}}$ must satisfy condition \textbf{kb}. Using condition \textbf{sym}, we conclude that for any $\selem\in\Delta^\Inter$, the axiom $\bigsqcap_{\conD\in\val(\selem)}\conD\ssb\conC$ is a tautology for all $\conC\in\FLAT({\mathscr{T}})$. As shown above, $\selem\in\conD^{\Inter}$ for all $\conD\in\val(\selem)$, and thus $\selem\in\conC$. Hence every individual of $\Inter$ is an instance of each concept of $\FLAT({\mathscr{T}})$ as required.
\qed

The previous lemma shows soundness of our decision algorithm. Conversely, completeness is shown by the following lemma.

\begin{lemma}\label{lemma:domcomplete}
Consider an \ALCIb{} TBox ${\mathscr{T}}$. If ${\mathscr{T}}$ is satisfiable, then its canonical domino set $\DO_{\mathscr{T}}$ is nonempty.
\end{lemma}
\proof
This is a straightforward consequence of Lemma~\ref{lemma:domgreatest}: given a model $\Inter$ of ${\mathscr{T}}$, the domino projection $\domp_{\parts(\FLAT(\mathscr{T}))}(\Inter)$ is nonempty and (by Lemma~\ref{lemma:domgreatest}) contained in $\DO_{\mathscr{T}}$. Hence $\DO_{\mathscr{T}}$ is nonempty.
\qed

\fancypicture{%
\begin{figure}[tbp]
\begin{center}
\includegraphics[width=0.8\textwidth]{figures/executiveAlgorithmCut}
\caption{The type-elimination algorithm at a glance.\label{fig:DominoAlg}}
\end{center}
\end{figure}
Figure~\ref{fig:DominoAlg} displays an executive summary of the overall algorithm.} We now are ready to establish our main result on checking TBox satisfiability and the complexity of the given algorithm:

\begin{theorem}\label{theo:domcorrect}
An \ALCIb{} TBox ${\mathscr{T}}$ is satisfiable iff its canonical domino set $\DO_{\mathscr{T}}$ is nonempty. Algorithm~\ref{def:builddoms} thus describes a decision procedure for satisfiability of \ALCIb{} TBoxes. Moreover, the algorithm runs in exponential time and hence is worst-case optimal.
\end{theorem}

\proof
The first proposition of the theorem is a direct consequence of Lemma~\ref{lemma:domsound}, Proposition~\ref{prop:flateq} (page~\pageref{prop:flateq}), and Lemma~\ref{lemma:domcomplete}.

For worst-case optimality, recall that \SHIQb{} is \ExpTime{}-complete (see \citealp{RKH:Jelia-08}, where \ExpTime{}-hardness already directly follows from the results by \citealp{schild91:correspondence}). Now, considering the presented algorithm, we find that the set $\mathscr{C} = \parts(\FLAT({\mathscr{T}}))$ is linearly bounded by the size of ${\mathscr{T}}$, whence the size of the set of all dominoes is exponentially bounded by $|{\mathscr{T}}|$. Applying the conditions \textbf{kb}, \textbf{ex}, and \textbf{uni} to obtain $\DO_0$ can be done by subsequently checking every domino, each check taking at most $O(|{\mathscr{T}}|)$ time, hence the overall time for that step is exponentially bounded. Now, consider the iterated application of the \textbf{delex}, \textbf{deluni}, and \textbf{sym} conditions. By the same argumentation as for \textbf{kb}, \textbf{ex}, and \textbf{uni}, one iteration takes exponential time. On the other hand, each iteration step reduces the domino set by at least one domino (otherwise, the termination criterion would be satisfied) which gives us a bound of exponentially many steps. Finally note that exponentially many exponentially long steps still yield a procedure that is overall exponentially bounded.
\qed

\section{Sets as Boolean Functions}\label{sec:boolfunc}

The algorithm of the previous section may seem to be of little practical use, since it requires computations on an exponentially large set of dominoes. The required computation steps, however, can also be accomplished with an indirect representation of the possible dominoes based on Boolean functions. Indeed, every propositional logic formula represents a set of propositional interpretations for which the function evaluates to $\true$. Using a suitable encoding, each propositional interpretation can be understood as a domino, and a propositional formula can represent a domino set.

As a representation of propositional formulae well-proven in other contexts, we use binary decision diagrams (BDDs). These data structures have been used to represent complex Boolean functions in model-checking (see, e.g., \citealp{burchetal}). A particular optimization of these structures are ordered BDDs (OBDDs) that use a dynamic precedence order of propositional variables to obtain compressed representations. We provide a first introduction to OBDDs below. A more detailed exposition and further literature pointers are given by \cite{huth}.

\subsection{Boolean Functions and Operations}
We first explain how sets can be represented by means of Boolean functions. This will enable us, given a fixed finite base set $S$, to represent every family of sets $\mathbb{S}\subseteq 2^S$ by a single Boolean function.

A \define{Boolean function} on a set $\Var$ of variables is a function $\varphi: 2^\Var \to \{\true,\false\}$. The underlying intuition is that $\varphi(V)$ computes the truth value of a Boolean formula based on the assumption that exactly the variables of $V$ are set to $\true$. A simple example are the functions $\tobf{\true}$ and $\tobf{\false}$, that map every input to $\true$ or $\false$, respectively. Another example are so-called \define{characteristic functions} of the form $\charf{v}$ for some $v\in\Var$, which are defined as $\charf{v}(V)\coloneqq\true$ iff $v\in V$.

Boolean functions over the same set of variables can be combined and modified in several ways. Especially, there are the obvious Boolean operators for negation, conjunction, disjunction, and implication. By slight abuse of notation, we will use the common (syntactic) operator symbols $\neg$, $\wedge$, $\vee$, and $\to$ to also represent such (semantic) operators on Boolean functions. Given, e.g., Boolean functions $\varphi$ and $\psi$, we find that $(\varphi\wedge\psi)(V)=\true$ iff $\varphi(V)=\true$ and $\psi(V)=\true$. Note that the result of the application of $\wedge$ results in another Boolean function, and is not to be understood as a syntactic logical formula.

Another operation on Boolean functions is existential quantification over a set of variables $V\subseteq\Var$, written as $\exists V.\varphi$ for some function $\varphi$. Given an input set $W\subseteq\Var$ of variables, we define $(\exists V.\varphi)(W)=\true$ iff \emph{there is some} $V'\subseteq V$ such that $\varphi(V'\cup (W\setminus V) )=\true$. In other words, there must be a way to set truth values of variables in $V$ such that $\varphi$ evaluates to $\true$. Universal quantification is defined analogously, and we thus have $\forall V.\varphi \coloneqq \neg\exists V.\neg\varphi$ as usual. Mark that our use of $\exists$ and $\forall$ overloads notation, and should not be confused with role restrictions in DL expressions.

\subsection{Ordered Binary Decision Diagrams}

Binary Decision Diagrams (BDDs), intuitively speaking, are a generalization of decision trees that allows for the reuse of nodes. Structurally, BDDs are directed acyclic graphs whose nodes are labeled by variables from some set $\Var$. The only exception are two \define{terminal} nodes that are labeled by $\true$ and $\false$, respectively. Every non-terminal node has two outgoing edges, corresponding to the two possible truth values of the variable.

\begin{definition}
A \define{BDD} is a tuple $\mathbb{O}=\tuple{N, \nroot, \ntrue, \nfalse, \low, \high, \Var, \lambda}$ where
\begin{iteMize}{$\bullet$}
\item $N$ is a finite set called \define{nodes},
\item $\nroot \in N$ is called the \define{root node},
\item $\ntrue,\nfalse \in N$ are called the \define{terminal nodes},
\item $\low,\high: N\setminus\{\ntrue,\nfalse\} \to N$ are two
\define{child functions} assigning to every non-terminal node a \define{low} and a \define{high} child node. Furthermore the graph obtained by iterated application has to be acyclic, i.e., for no node $n$ exists a sequence of applications of $\low$ and $\high$ resulting in $n$ again.
\item $\Var$ is a finite set of \define{variables}.
\item $\lambda: N\setminus\{\ntrue,\nfalse\} \to \Var$ is the \define{labeling function} assigning to every non-terminal node a variable from $\Var$.
\end{iteMize}
\end{definition}

OBDDs are a particular realization of BDDs where a certain ordering is imposed on variables to achieve more efficient representations. We will not require to consider the background of this optimization in here. Every BDD based on a variable set $\Var=\{x_1,\ldots,x_n\}$ represents an $n$-ary Boolean function $\varphi:2^\Var\to \{\true,\false\}$.

\begin{definition}\label{def:obddsem}
Given a BDD $\mathbb{O}=\tuple{N,\nroot,\ntrue,\nfalse,\low,\high,\Var,\lambda}$ the Boolean function $\varphi_\mathbb{O}:2^\Var\to \{\true,\false\}$ is defined recursively as follows:
$$\varphi_\mathbb{O}\coloneqq \varphi_{\nroot} \qquad
\varphi_{\ntrue} = \tobf{\true}  \qquad \varphi_{\nfalse} = \tobf{\false}$$
$$\varphi_{n} =
\Big(\neg\charf{\lambda(n)}\wedge\varphi_{\low(n)}\Big) \vee \Big(\charf{\lambda(n)}\wedge\varphi_{\high(n)}\Big) \quad\text{ for } n\in N\setminus\{\ntrue,\nfalse\}$$
\end{definition}

In other words, the value $\varphi(V)$ for some $V\subseteq \Var$ is determined by traversing the BDD, starting from the root node: at a node labeled with $v\in\Var$, the evaluation proceeds with the node connected by the $\high$-edge if $v\in V$, and with the node connected by the $\low$-edge otherwise. If a terminal node is reached, its label is returned as a result.

%

BDDs for some Boolean formulas might be exponentially large in general (compared to $|\Var|$), but often there is a representation which allows for BDDs of manageable size. Finding the optimal representation is NP-complete, but heuristics have shown to yield good approximate solutions \citep{Wegener04bdds}. Hence (O)BDDs are often conceived as efficiently compressed representations of Boolean functions. In addition, many operations on Boolean functions -- such as the aforementioned negation, conjunction, disjunction, implication as well as propositional quantification -- can be performed directly on the corresponding OBDDs by fast algorithms.

\subsection{Translating Dominos into Boolean Functions}

To apply the above machinery to DL reasoning, consider a flattened \ALCIb{} TBox $\mathscr{T}=\FLAT(\mathscr{T})$. A set of propositional variables $\Var$ is defined as $\Var \coloneqq \lang{R} \cup \big(\parts(\mathscr{T}) \times \{1,2\}\big)$. We thus obtain a bijection between dominoes over the set $\parts(\mathscr{T})$ and sets $V\subseteq\Var$ given by $\tuple{\mathscr{A},\mathscr{R},\mathscr{B}} \mapsto (\mathscr{A}\times\{1\}) \cup \mathscr{R} \cup (\mathscr{B}\times\{2\})$. Hence, any Boolean function over $\Var$ represents a domino set as the collection of all variable sets for which it evaluates to $\true$. \fancypicture{Figure~\ref{fig:DominoEnc} illustrates this correspondence.
\begin{figure}[bht]
\begin{center}
\includegraphics[width=0.6\textwidth]{figures/executiveOBDDencoding}
\caption{Encoding of Domino sets into Boolean functions.\label{fig:DominoEnc}}
\end{center}
\end{figure}}
We can use this observation to rephrase the construction of $\DO_{\mathscr{T}}$ in Algorithm~\ref{def:builddoms} into an equivalent construction of a function $\tobf{\mathscr{T}}$.

We first represent DL concepts $\conC$ and role expressions $\rolU$ by characteristic Boolean functions over $\Var$ as follows.

\[ \tobf{\conC} \coloneqq \left\{\begin{array}{ll}
        \neg\tobf{\conD}                 & \text{ if } \conC = \neg\conD \\
        \tobf{\conD} \wedge \tobf{\conE} & \text{ if } \conC = \conD\sqcap\conE \\
        \tobf{\conD} \vee \tobf{\conE}   & \text{ if } \conC = \conD\sqcup\conE \\
        \charf{\tuple{\conC,1}}           & \text{ if } \conC\in\parts(\mathscr{T}) \\
\end{array}\right.
\qquad
\tobf{\rolU}\coloneqq\left\{\begin{array}{ll}
        \neg \tobf{\rolV}                & \text{ if } \rolU = \neg\rolV \\
        \tobf{\rolV} \wedge \tobf{\rolW} & \text{ if } \rolU = \rolV\sqcap\rolW \\
        \tobf{\rolV} \vee \tobf{\rolW}   & \text{ if } \rolU = \rolV\sqcup\rolW \\
        \charf{\rolU}                     & \text{ if } \rolU\in\lang{R} \\
\end{array}\right.
\]

\newcommand{\ghost}[1]{\raisebox{0pt}[0pt][0pt]{\makebox[0pt][l]{#1}}}

\begin{algorithm}[t]
  \begin{algorithmic} [1]
    \REQUIRE \hspace*{1.5ex} ${\mathscr{T}}$ an \ALCIb{} TBox, $\mathscr{C}=\parts(\FLAT({\mathscr{T}}))$
    \ENSURE the canonical domino set of ${\mathscr{T}}$, represented as Boolean function $\tobf{\mathscr{T}}$\\

    \STATE \makebox[0pt][l]{$\varphi^\mathbf{kb}$}\phantom{$\tobf{\mathscr{T}}_0$}$\displaystyle{}:=\bigwedge_{\conC\in\mathscr{T}}\tobf{\conC}$
    \STATE \makebox[0pt][l]{$\varphi^\mathbf{uni}$}\phantom{$\tobf{\mathscr{T}}_0$}$\displaystyle{}:= \bigwedge_{\ghost{$\scriptstyle\forall \rolU.\conC \in \parts(\mathscr{T})$}\phantom{Hack}}\charf{\tuple{\forall\rolU.\conC, 1}} \wedge \tobf{\rolU} \to \charf{\tuple{\conC, 2}}$
    \STATE \makebox[0pt][l]{$\varphi^\mathbf{ex}$}\phantom{$\tobf{\mathscr{T}}_0$}$\displaystyle{}:=\bigwedge_{\ghost{$\scriptstyle\exists \rolU.\conC \in \parts(\mathscr{T})$}\phantom{Hack}}\charf{\tuple{\conC,2}}
    \wedge \tobf{\rolU} \to \charf{\tuple{\exists\rolU.\conC,1}}$\\[1ex]
    \STATE $\tobf{\mathscr{T}}_0 \coloneqq \varphi^\mathbf{kb} \wedge \varphi^\mathbf{uni} \wedge \varphi^\mathbf{ex}$
    \STATE i := 0
    \REPEAT
    \STATE i := i+1\\[1ex]
    \STATE \makebox[0pt][l]{$\varphi^{\mathbf{delex}}_i$}\phantom{$\varphi^{\mathbf{sym}}_i(V)$}$\displaystyle{}:=\bigwedge_{\ghost{$\scriptstyle\exists\rolU.\conC\in\parts(\mathscr{T})$}\phantom{Hack}} \charf{\tuple{\exists\rolU.\conC,1}} \to
    \exists \big(\mathbf{R}\cup \mathscr{C}\!\times\!\{2\}\big).\big(\tobf{\mathscr{T}}_{i-1}
    \wedge \tobf{\rolU} \wedge \charf{\tuple{\conC,2}}\big)$
    \STATE
    \makebox[0pt][l]{$\varphi^{\mathbf{deluni}}_i$}\phantom{$\varphi^{\mathbf{sym}}_i(V)$}$\displaystyle{}:=\bigwedge_{\ghost{$\scriptstyle\forall\rolU.\conC\in\parts(\mathscr{T})$}\phantom{Hack}}\charf{\tuple{\forall\rolU.\conC,1}} \to \neg\exists \big(\mathbf{R}\cup \mathscr{C}\!\times\!\{2\}\big).\big(\tobf{\mathscr{T}}_{i-1}
    \wedge \tobf{\rolU} \wedge \neg\charf{\tuple{\conC,2}}\big)$
    \STATE
    $\displaystyle\varphi^{\mathbf{sym}}_i(V):= \tobf{\mathscr{T}}_{i-1}\Big(  \big\{ \tuple{\conD,1}\mid \tuple{\conD,2}\in V \big\} \cup
    \big\{ \Inv(\rolR)\mid \rolR\in V \big\} \cup \big\{ \tuple{\conD,2}\mid \tuple{\conD,1}\in V\big \} \Big)$\\[1ex]
    \STATE
    $\tobf{\mathscr{T}}_{i} \coloneqq \tobf{\mathscr{T}}_{i-1} \wedge \varphi^{\mathbf{delex}}_{i} \wedge \varphi^{\mathbf{deluni}}_{i} \wedge \varphi^{\mathbf{sym}}_{i}$
     \UNTIL $\tobf{\mathscr{T}}_{i} \equiv \tobf{\mathscr{T}}_{i-1}$
     \STATE $\tobf{\mathscr{T}}\coloneqq\tobf{\mathscr{T}}_i$
     \RETURN $\tobf{\mathscr{T}}$
  \end{algorithmic}
  \caption{Computing the boolean representation $\tobf{\mathscr{T}}$ of the canonical domino set $\DO_{\mathscr{T}}$ of a TBox}
  \label{def:builddomsbf}
\end{algorithm}

$\left.\right.$\\
We can now define a decision procedure based on Boolean functions, as displayed in Algorithm~\ref{def:builddomsbf}. This algorithm is an accurate translation of Algorithm~\ref{def:builddoms}, where the intermediate Boolean functions $\varphi^\mathbf{kb},\varphi^\mathbf{ex},\varphi^\mathbf{uni},\varphi_i^\mathbf{delex},\varphi_i^\mathbf{deluni},\varphi_i^\mathbf{sym}$ represent domino sets containing all dominoes satisfying the respective conditions from Algorithm~\ref{def:builddoms}. By computing their conjunction with each other (and, for the latter three, with the Boolean function representing the domino set from the previous iteration) we intersect the respective domino sets which results in their successive pruning as described in Algorithm~\ref{def:builddoms}.
The algorithm is a correct procedure for checking consistency of \ALCIb{} TBoxes as unsatisfiability of $\mathscr{T}$ coincides with $\tobf{\mathscr{T}}\equiv \false$. Note that all necessary computation steps can indeed be implemented algorithmically: Any Boolean function can be evaluated for a fixed variable input $V$, and equality of two functions can (naively) be checked by comparing the results for all possible input sets (which are finitely many since $\Var$ is finite). The algorithm terminates since the sequence is decreasing w.r.t.\ $\{V\mid \tobf{\mathscr{T}}_i(V)=\true\}$, and since there are only finitely many Boolean functions over $\Var$.


\begin{proposition}\label{prop:obdddo}
For any \ALCIb{} TBox $\mathscr{T}$ and variable set $V\in\Var$ as above, we find that $\tobf{\mathscr{T}}(V)=\true$ iff $V$ represents a domino in $\DO_{\mathscr{T}}$ as defined in Definition~\ref{def:builddoms}.
\end{proposition}
\proof
It is easy to see that the Boolean operations used in constructing $\tobf{\mathscr{T}}$ directly correspond to the set operations in Definition~\ref{def:builddoms}, such that $\tobf{\mathscr{T}}(V)=\true$ iff $V$ represents a domino in $\DO_{\kb}$.
\qed

All required operations and checks are provided by standard OBDD implementations, and thus can be realized in practice. 

\medskip

In the remainder of this section, we illustrate the above algorithm by an extended example to which we will also come back to explain the later extensions of the inference algorithm. Therefore, consider the following \ALCIb{} knowledge base $\mathscr{KB}$.
\[\begin{array}{rll}
\dlname{PhDStudent} & \ssb & \exists \dlname{has}.\dlname{Diploma}\\
\dlname{Diploma} & \ssb & \forall \dlname{has}^-.\dlname{Graduate}\\
\dlname{Diploma} \sqcap \dlname{Graduate} & \ssb & \bot\\
\dlname{Diploma}(\dlname{laureus}) & & \dlname{PhDStudent}(\dlname{laureus})\\
\end{array}\]
For now, we are only interested in the terminological axioms, the consistency of which we would like to establish. 
As a first transformation step, all TBox axioms are transformed into the following universally valid concepts in negation normal form:
\[ \neg\dlname{PhDStudent} \sqcup \exists \dlname{has}.\dlname{Diploma}\quad
  \neg\dlname{Diploma} \sqcup \forall \dlname{has}^-.\dlname{Graduate}\quad
  \neg\dlname{Diploma} \sqcup \neg\dlname{Graduate} \]
The flattening step can be skipped since all concepts are already flat. Now the relevant concept expressions for describing dominoes are given by the set \[\parts(\mathscr{T})=\{\exists \dlname{has}.\dlname{Diploma}, \forall \dlname{has}^-\!.\dlname{Graduate}, \dlname{Diploma}, \dlname{Graduate}, \dlname{PhDStudent}\}.\] We thus obtain the following set $\Var$ of Boolean variables (although $\Var$ is just a set, our presentation follows the domino intuition):

\[
\begin{array}{|l|l|l|}\hline
\tuple{\exists \dlname{has}.\dlname{Diploma},1} & \dlname{has} & \tuple{\exists \dlname{has}.\dlname{Diploma},2}\\
\tuple{\forall \dlname{has}^-.\dlname{Graduate},1} & \dlname{has}^- & \tuple{\forall \dlname{has}^-.\dlname{Graduate},2} \\
\tuple{\dlname{Diploma},1} & & \tuple{\dlname{Diploma},2} \\
\tuple{\dlname{Graduate},1} & & \tuple{\dlname{Graduate},2} \\
\tuple{\dlname{PhDStudent},1} & & \tuple{\dlname{PhDStudent},2}
\\\hline
\end{array}
\]

\begin{figure}[t]
%
%
\includegraphics[width=0.8\textwidth]{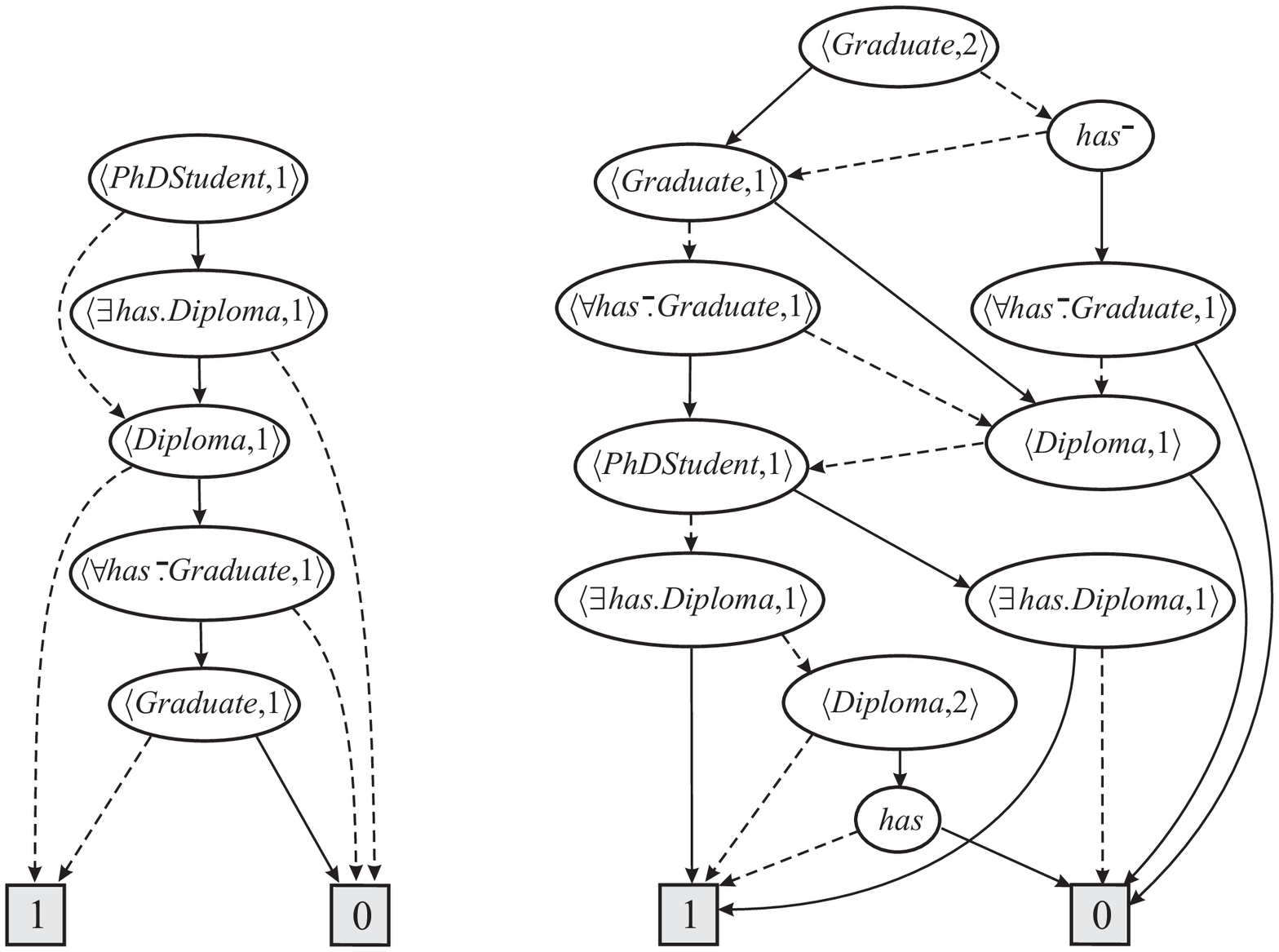}
\hfill~ \caption{OBDDs arising when processing the terminology of $\mathscr{KB}$; following traditional BDD notation, solid arrows indicate $\high$ successors, dashed arrows indicate $\low$ successors, and the topmost node is the root\label{fig:obdd1}}
\end{figure}%

\vspace{2ex}

We are now ready to construct the OBDDs as described. Figure~\ref{fig:obdd1} (left) displays an OBDD corresponding to the following Boolean function:
\[
\begin{array}{rl}
\varphi^\mathbf{kb} \coloneqq & (\neg\charf{\tuple{\dlname{PhDStudent},1}}\vee
\charf{\tuple{\exists \dlname{has}.\dlname{Diploma},1}})\\
 & \wedge (\neg \charf{\tuple{\dlname{Diploma},1}}
\vee\charf{\tuple{\forall
\dlname{has}^-.\dlname{Graduate},1}})\\
 & \wedge(\neg \charf{\tuple{\dlname{Diploma},1}}
\vee\neg\charf{\tuple{\dlname{Graduate},1}})
\end{array}
\]
and Fig.~\ref{fig:obdd1} (right) shows the OBDD representing the function $\tobf{\mathscr{T}}_0$ obtained from $\varphi^{\mathbf{kb}}$ by conjunctively adding
$$\begin{array}{rl} \varphi^\mathbf{ex}\quad =\quad &
 \neg\charf{\tuple{\dlname{Diploma},2}}\vee \neg\charf{\dlname{has}} \vee \charf{\tuple{\exists\dlname{has}.\dlname{Diploma},1}}
\quad\text{and}\\
\varphi^\mathbf{uni}\quad =\quad &
  \neg\charf{\tuple{\forall\dlname{has}^-.\dlname{Graduate},1}}\vee \neg\charf{\dlname{has}^-} \vee \charf{\tuple{\dlname{Graduate},2}}.\\
\end{array}
$$
\begin{figure}[b]
\includegraphics[width=0.78\textwidth]{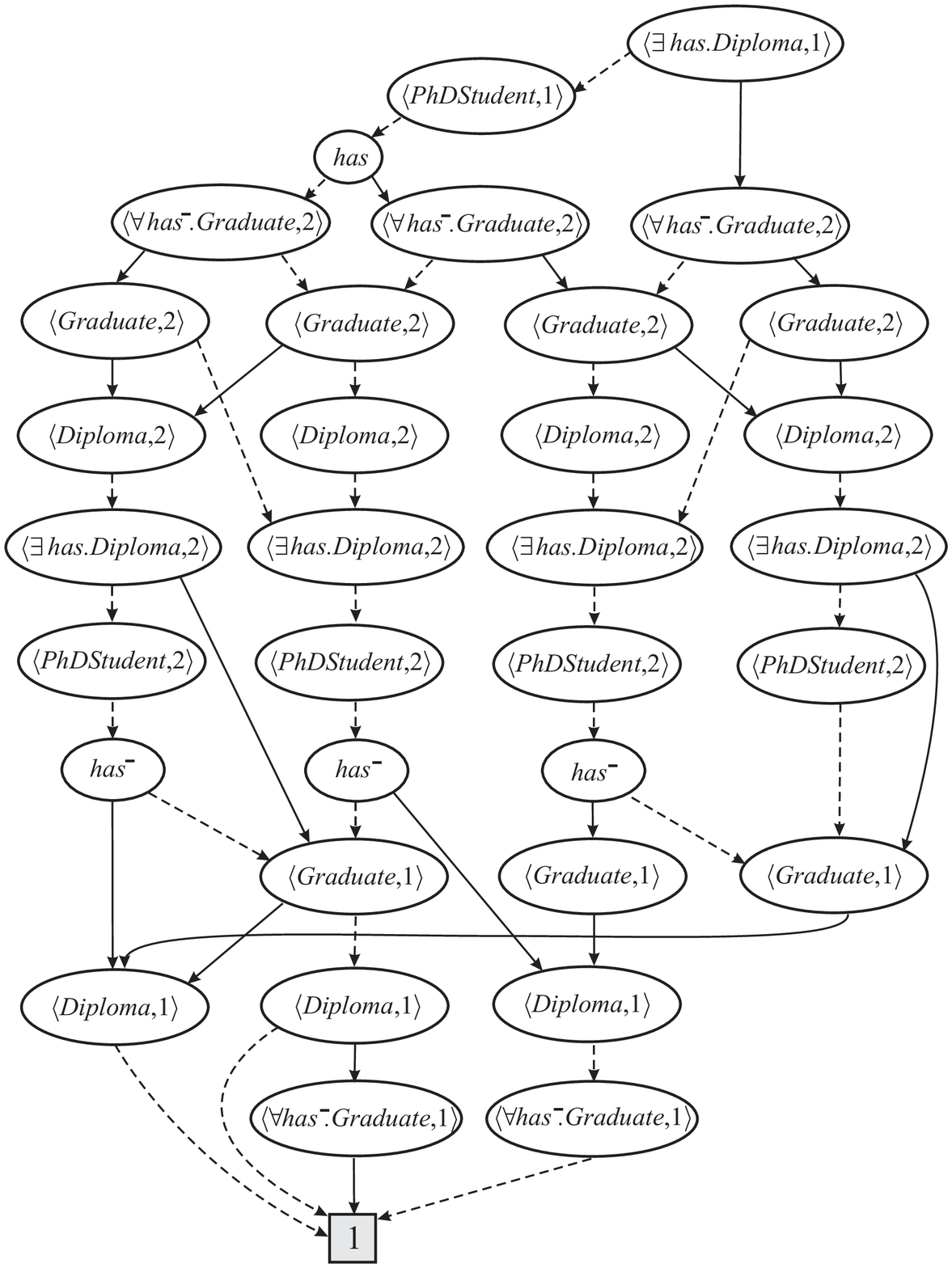}
\caption{Final OBDD obtained when processing $\mathscr{KB}$, using notation as in Fig.~\ref{fig:obdd1}; arrows to the $0$ node have been omitted for better readability\label{fig:obdd3}}
\end{figure}
\noindent Then, after the first iteration of the algorithm, we arrive at an OBDD representing $\tobf{\mathscr{T}}_1$ which is displayed in Fig.~\ref{fig:obdd3}. This OBDD turns out to be the final result $\tobf{\mathscr{T}}$. The input TBox is derived to be consistent since there is a path from the root node to $1$.
\pagebreak

\section{Reasoning with ABox and DL-Safe Rules via Disjunctive Datalog}\label{sec:abox}

The above algorithm does not yet take any assertional information about individuals into account, nor does it cover DL-safe rules. The proof of Theorem~\ref{theo:domcorrect} hinges upon the fact that the constructed domino set $\DO_{\mathscr{T}}$ induces a model of the terminology $\mathscr{T}$, and Lemma~\ref{lemma:domgreatest} states that this is indeed the \emph{greatest} model in a certain sense. This provides some first intuition of the problems arising when ABoxes are to be added to the knowledge base: \ALCIb{} knowledge bases with ABoxes do generally not have a greatest model.

We thus employ \emph{disjunctive Datalog} (see \citealp{disjDatalog}) as a paradigm that allows us to incorporate ABoxes into the reasoning process. The basic idea is to forge a Datalog program that -- depending on two given individuals $a$ and $b$ -- describes possible dominoes that may connect $a$ and $b$ in models of the knowledge base. There might be various, irreconcilable such dominoes in different models, but disjunctive Datalog supports such choice since it admits multiple minimal models. As long as the knowledge base has some model, there is at least one possible domino for every pair of individuals (possibly without connecting roles) -- only if this is not the case, the Datalog program will infer a contradiction. Another reason for choosing disjunctive Datalog is that it allows for the straightforward incorporation of DL-safe rules.


We use the OBDD computed from  the terminology as a kind of pre-compiled version of the relevant terminological information. ABox information is then considered as an incomplete specification of dominoes that must be accepted by the OBDD, and the Datalog program simulates the OBDD's evaluation for each of those.


\begin{definition}\label{def:datalog}
Consider an extended \ALCIb{} knowledge base $\kb=\tuple{\mathscr{T},\mathscr{P}}$, and an OBDD
$\mathbb{O}=\tuple{N, \nroot, \ntrue, \nfalse, \low, \high, \Var, \lambda}$ that represents the function $\tobf{\mathscr{T}}$ as defined by Algorithm~\ref{def:builddomsbf}.
A disjunctive Datalog program $\Prog(\kb)$ is defined as follows. $\Prog(\kb)$ uses the following predicates:
\begin{iteMize}{$\bullet$}
\item a unary predicate $S\!_\conC$ for every concept expression $\conC\in\parts(\FLAT(\mathscr{T}))$,
\item a binary predicate $S\!_\rolR$ for every atomic role $\rolR\in\rolnames$,
\item a binary predicate $A_n$ for every OBDD node $n\in N$,
\item the equality predicate $\approx$.
\end{iteMize}
The constants in $\Prog(\kb)$ are the individual names used in $\mathscr{P}$. The disjunctive Datalog rules of $\Prog(\kbrb)$ are defined as follows:\footnote{Note that we use disjunctive Datalog with equality. However, every disjunctive Datalog program with equality can be reduced to one without equality in linear time, as equality can be axiomatized (see, e.g., \citealp{Fitting}).}
\begin{enumerate}[(1)]
\item\label{item:rb-rules}
For every DL-safe rule $B\to H$ from $\rb$, $\Prog(\kb)$ contains the rule obtained from $B\to H$ by replacing all $\conC(x)$ by $S\!_\conC(x)$ and all $\rolR(x,y)$ by $S\!_\rolR(x,y)$.
\item\label{item:topbot-rules}
$\Prog(\kb)$ contains rules $\to A_{\nroot}(x,y)$ and $A_{\nfalse}(x,y)\to$.
\item\label{item:cone-rules}
If $n\in N$ with $\lambda(n) = \tuple{\conC,1}$ then $\Prog(\kb)$ contains rules\\
$S\!_\conC(x) \wedge A_n(x,y) \to A_{\high(n)}(x,y)$ and $A_n(x,y) \to A_{\low(n)}(x,y) \vee S\!_\conC(x)$.
\item\label{item:ctwo-rules}
If $n\in N$ with $\lambda(n) = \tuple{\conC,2}$ then $\Prog(\kb)$ contains rules\\
$S\!_\conC(y) \wedge A_n(x,y) \to A_{\high(n)}(x,y)$ and $A_n(x,y) \to A_{\low(n)}(x,y) \vee S\!_\conC(y)$.
\item\label{item:role-rules}
If $n\in N$ with $\lambda(n) = \rolR$ for some $\rolR\in\rolnames$ then $\Prog(\kb)$ contains rules\\
$S\!_\rolR(x,y) \wedge A_n(x,y) \to A_{\high(n)}(x,y)$ and $A_n(x,y) \to A_{\low(n)}(x,y) \vee S\!_\rolR(x,y)$.
\item\label{item:rinv-rules}
If $n\in N$ with $\lambda(n) = \rolR^-$ for some $\rolR\in \rolnames$ then $\Prog(\kb)$ contains rules\\
$S\!_\rolR(y,x) \wedge A_n(x,y) \to A_{\high(n)}(x,y)$ and $A_n(x,y) \to A_{\low(n)}(x,y) \vee S\!_\rolR(y,x)$.
\end{enumerate}
\end{definition}

\fancypicture{\begin{figure}[bht]
\begin{center}
\includegraphics[width=1\textwidth]{figures/executiveOBDDtoDatalog}
\caption{Translation of an OBDD into a disjunctive Datalog program.\label{fig:DominoDatalog}}
\end{center}
\end{figure}
Figure~\ref{fig:DominoDatalog} visualizes the translation strategy.} Note that the arity of predicates in $\Prog(\kb)$ is bounded by $2$. Hence, the number of ground atoms is quadratic with respect to the number of constants (individual names), whence the worst-case complexity for satisfiability checking is \NP{} w.r.t.\ the number of individuals (and especially w.r.t.\ the number of facts), as opposed to the \NExpTime{} complexity of disjunctive Datalog in general (\citealp{DBLP:journals/csur/DantsinEGV01}). Note that, of course, $\Prog(\kb)$ may still be exponential in the size of $\kb$ in the worst case: $\Prog(\kb)$ is linear in the size of the underlying OBDD which in turn may have exponential size compared to the set of propositional variables used in the represented Boolean functions. Finally the number of these variables is linearly bounded by the size of $\kb$. It remains to show the correctness of the Datalog translation.

\begin{lemma}\label{lemma:ddsound}
Given an extended \ALCIb{} knowledge base \kb{} such that $\Inter$ is a model of $\kb$, there is a model $\Jnter$ of $\Prog(\kb)$ such that
\begin{iteMize}{$\bullet$}
\item
$\Inter\models \conC(a)$ iff $\Jnter\models S_{\!\conC}(a)$,
\item
$\Inter\models \rolR(a,b)$ iff $\Jnter\models S_{\!\rolR}(a,b)$, and
\item
$\Inter\models a\approx b$ iff $\Jnter\models a\approx b$.
\end{iteMize}
for any $a,b\in\indnames$, $\conC\in\connames$, and $\rolR\in\rolnames$.
\end{lemma}
\proof
Let $\kb=\tuple{\mathscr{T},\mathscr{P}}$. We define an interpretation $\Jnter$ of $\Prog(\kb)$. The domain of $\Jnter$ contains the named individuals from $\Inter$, i.e., $\Delta^\Jnter=\{a^\Inter \mid a\in\indnames\}$. For individuals $a$, we set $a^\Jnter \defeq a^\Inter$. The interpretation of predicate symbols is now defined as follows (note that $A_n^\Jnter$ is defined inductively on the path length from $\nroot$ to $n$):
\begin{iteMize}{$\bullet$}
\item $\delem\in S\!_\conC^\Jnter$ iff $\delem \in \conC^\Inter$
\item $\tuple{\delem_1,\delem_2}\in S\!_\rolR^\Jnter$ iff $\tuple{\delem_1,\delem_2}\in \rolR^\Inter$
\item $\tuple{\delem_1,\delem_2}\in A_{\nroot}^{\Jnter}$ for all $\delem_1,\delem_2\in\Delta^{\Jnter}$
\item $\tuple{\delem_1,\delem_2}\in A_n^\Jnter$ for $n\neq\nroot$ if there is a node $n'$ such that $\tuple{\delem_1,\delem_2}\in A_{n'}^\Jnter$, and one of the following is the case:
\begin{iteMize}{$-$}
\item $\lambda(n')=\tuple{\conC,i}$, for some $i\in\{1,2\}$, and $n = \low(n')$ and $\delem_i \not\in \conC^\Inter$
\item $\lambda(n')=\tuple{\conC,i}$, for some $i\in\{1,2\}$, and $n = \high(n')$ and $\delem_i \in \conC^\Inter$
\item $\lambda(n')=\rolR$ and $n = \low(n')$ and $\tuple{\delem_1,\delem_2} \not\in \rolR^\Inter$
\item $\lambda(n')=\rolR$ and $n = \high(n')$ and $\tuple{\delem_1,\delem_2} \in \rolR^\Inter$
\end{iteMize}
\end{iteMize}
Mark that, in the last two items, $\rolR$ is any role expression from $\Var$, i.e., a role name or its inverse. Also note that due to the acyclicity of $\mathbb{O}$, the interpretation of the $A$-predicates is indeed well-defined. We now show that $\Jnter$ is a model of $\Prog(\kb)$. To this end, first note that the extensions of predicates $S\!_\conC$ and $S\!_\rolR$ in $\Jnter$ were defined to coincide with the extensions of $\conC$ and $\rolR$ on the named individuals of $\Inter$. Since $\Inter$ satisfies $\mathscr{P}$, all rules introduced in item \eqref{item:rb-rules} of Definition~\ref{def:datalog} are satisfied by $\Jnter$. The restriction of DL-safe rules to named individuals can be discarded here since $\Delta^\Jnter$ contains only named individuals from $\Delta^\Inter$.

Similarly, we find that the rules of cases \eqref{item:cone-rules}--\eqref{item:rinv-rules} are satisfied by $\Jnter$. Consider the first rule of \eqref{item:cone-rules}, $S\!_\conC(x) \wedge A_n(x,y) \to A_{\high(n)}(x,y)$, and assume that $\delem_1\in S_{\conC}^\Jnter$ and $\tuple{\delem_1,\delem_2}\in A_n^{\Jnter}$. Thus $\delem_1\in\conC^{\Inter}$. Using the preconditions of \eqref{item:cone-rules} and the definition of $\Jnter$, we conclude that $\tuple{\delem_1,\delem_2}\in A_{\high(n)}^{\Jnter}$. The second rule of case \eqref{item:cone-rules} covers the analogous negative case. All other cases can be treated similarly.

Finally, for case \eqref{item:topbot-rules}, we need to show that $A_{\nfalse}^{\Jnter} = \emptyset$. For that, we first explicate the correspondence between domain elements of $\Inter$ and sets of variables of $\mathbb{O}$. Given elements $\delem_1,\delem_2 \in \Delta^\Inter$ we define $V_{\delem_1,\delem_2}\defeq \{\tuple{\conC,n}\mid \conC\in \parts(\FLAT(\mathscr{T})),\delem_n\in\conC^\Inter\} \cup\{\rolR\mid\tuple{\delem_1,\delem_2}\in\rolR^\Inter\}$, the set of variables corresponding to the $\Inter$-domino between $\delem_1$ and $\delem_2$.


Now $A_{\nfalse}^{\Jnter} = \emptyset$ clearly is a consequence of the following claim: for all $\delem_1,\delem_2\in\Delta^\Inter$ and all $n\in N$, we find that $\tuple{\delem_1,\delem_2}\in A_n$ implies $\varphi_n(V_{\delem_1,\delem_2})=\true$ (using the notation of Definition~\ref{def:obddsem}). The proof proceeds by induction. For the case $n=\nroot$, we find that $\varphi_{\nroot} = \tobf{\mathscr{T}}$. Since $V_{\delem_1,\delem_2}$ represents a domino of $\Inter$, the claim thus follows by combining Proposition~\ref{prop:obdddo} and Lemma~\ref{lemma:domgreatest}.

For the induction step, let $n$ be a node such that $\tuple{\delem_1,\delem_2}\in A_n$ follows from the inductive definition of $\Jnter$ based on some predecessor node $n'$ for which the claim has already been established. Note that $n'$ may not be unique. The cases in the definition of $\Jnter$ must be considered individually. Thus assume $n'$, $n$, and $\delem_1$ satisfy the first case, and that $\tuple{\delem_1,\delem_2}\in A_n$. By induction hypothesis, $\varphi_{n'}(V_{\delem_1,\delem_2})=\true$, and by Definition~\ref{def:obddsem} the given case yields $\varphi_{n}(V_{\delem_1,\delem_2})=\true$ as well. The other cases are similar.
\qed

\begin{lemma}\label{lemma:ddcomp}
Given an \ALCIb{} knowledge base \kb{} such that $\Jnter$ is a model of $\Prog(\kb)$, there is a model $\Inter$ of $\kb$ such that
\begin{iteMize}{$\bullet$}
\item
$\Inter\models \conC(a)$ iff $\Jnter\models S_{\!\conC}(a)$,
\item
$\Inter\models \rolR(a,b)$ iff $\Jnter\models S_{\!\rolR}(a,b)$, and
\item
$\Inter\models a\approx b$ iff $\Jnter\models a\approx b$,
\end{iteMize}
for any $a,b\in\indnames$, $\conC\in\connames$, and $\rolR\in\rolnames$.
\end{lemma}
\proof
Let $\kb=\tuple{\mathscr{T},\mathscr{P}}$.
%
We construct an interpretation $\Inter$ whose domain $\Delta^\Inter$ consists of all sequences starting with an individual name followed by a (possibly empty) sequence of dominoes from $\DO_{\mathscr{T}}$ such that, for every $\selem\in\Delta^\Inter$,
\begin{iteMize}{$\bullet$}
\item if $\selem$ begins with $a\tuple{\mathscr{A},\mathscr{R},\mathscr{B}}$, then $\{\conC \mid \conC \in \parts(\FLAT(\mathscr{T})), a^\Jnter \in S\!_\conC^\Jnter\} = \mathscr{A}$, and
\item if $\selem$ contains subsequent letters $\tuple{\mathscr{A},\mathscr{R},\mathscr{B}}$ and $\tuple{\mathscr{A}',\mathscr{R}',\mathscr{B}'}$, then $\mathscr{B}=\mathscr{A}'$.
\end{iteMize}
For a sequence $\selem = a\tuple{\mathscr{A}_1,\mathscr{R}_1,\mathscr{A}_2} \tuple{\mathscr{A}_2,\mathscr{R}_2,\mathscr{A}_3} \ldots \tuple{\mathscr{A}_{i-1},\mathscr{R}_{i-1},\mathscr{A}_i}$, we define $\val(\selem)\defeq \mathscr{A}_i$, whereas for a $\selem = a$ we define $\val(\selem)\defeq\{\conC \mid \conC \in
\parts(\FLAT(\mathscr{T})), a^\Jnter \in S\!_\conC^\Jnter\}$. Now
the mappings of $\Inter$ are defined as follows:
\begin{iteMize}{$\bullet$}
\item for $a\in\indnames$, we have $a^\Inter\defeq a$,
\item for $\conA\in\connames$, we have $\selem \in \conA^\Inter$ iff $\conA\in \val(\selem)$,
\item for $\rolR\in\rolnames$, we have $\tuple{\selem_1,\selem_2}\in\rolR^\Inter$ if one of the following holds
    \begin{iteMize}{$-$}
    \item $\selem_1= a \in \indnames$ and $\selem_2= b \in \indnames$ and $\tuple{a,b} \in S\!_\rolR^\Jnter$, or
    \item $\selem_2= \selem_1 \tuple{\mathscr{A},\mathscr{R},\mathscr{B}} \text{ with }\rolR \in \mathscr{R}$, or
    \item $\selem_1= \selem_2 \tuple{\mathscr{A},\mathscr{R},\mathscr{B}} \text{ with } \Inv(\rolR) \in \mathscr{R}$.
    \end{iteMize}
\end{iteMize}
Thus, intuitively, $\Inter$ is constructed by extracting the named individuals as well their concept (and mutual role) memberships from $\Jnter$, and appending an appropriate domino-constructed tree model to each of those named individuals. We proceed by showing that $\Inter$ is indeed a model of $\kb$.

First note that the definition of $\Inter$ ensures that, for all individual names $a,b\in\indnames$, we indeed have $\Inter\models \conC(a)$ iff $\Jnter\models S_{\conC}(a)$, $\Inter\models \rolR(a,b)$ iff $\Jnter\models S_{\rolR}(a,b)$, and $\Inter\models a\approx b$ iff $\Jnter\models a\approx b$. Therefore, the validity of the rules introduced via case \eqref{item:rb-rules} ensures that $\Inter$ is a model of $\mathscr{P}$.

For showing that the TBox is also satisfied, we begin with the following auxiliary observation: for every two individual names $a, b\in\indnames$, and $\mathscr{R}_{ab}\defeq\{\rolR\mid \tuple{a^\Jnter,b^\Jnter}\in S\!_\rolR^\Jnter\}\cup\{\Inv(\rolR)\mid \tuple{b^\Jnter,a^\Jnter}\in S\!_\rolR^\Jnter\}$, the domino $\tuple{\val(a),\mathscr{R}_{ab},\val(b)}$ is contained in $\DO_{\mathscr{T}}$ (Claim \dag). Using Proposition~\ref{prop:obdddo}, it suffices to show that the Boolean function $\tobf{\mathscr{T}}$ if applied to $V_{a,b}\defeq\{\val(a)\times\{1\} \cup\mathscr{R}_{ab}\cup\val(b)\times\{2\}\}$ yields $\true$. Since $\tobf{\mathscr{T}}=\varphi_{\nroot}$, this is obtained by showing the following: for any $a,b\in\indnames$, we find that $\tuple{a^\Jnter,b^\Jnter} \in A_n^\Jnter$ implies $\varphi_n(V_{a,b})=\true$. Indeed, (\dag) follows since we have $\tuple{a^\Jnter,b^\Jnter} \in A_{\nroot}^\Jnter$ due to the first rule of \eqref{item:topbot-rules} in Definition~\ref{def:datalog}. We proceed by induction, starting at the leafs of the OBDD. The case $\tuple{a,b} \in A_{\ntrue}^\Inter$ is immediate, and $\tuple{a,b} \in A_{\nfalse}^\Inter$ is excluded by the second rule of \eqref{item:topbot-rules}. For the induction step, consider nodes $n, n'\in N$ such that either $\lambda(n)\in V_{a,b}$ and $n'=\high(n)$, or $\lambda(n)\notin V_{a,b}$ and $n'=\low(n)$. We assume that $\tuple{a^\Jnter,b^\Jnter} \in A_n^\Jnter$, and, by induction, that the claim holds for $n'$. If $\lambda_n=\tuple{\conC,1}$, then one of the rules of case \eqref{item:cone-rules} applies to $a^\Jnter$ and $b^\Jnter$. In both cases, we can infer $\tuple{a^\Jnter,b^\Jnter} \in A_{n'}^\Jnter$, and hence $\varphi_{n'}(V_{a,b})=\true$. Together with the assumptions for this case, Definition~\ref{def:obddsem} implies $\varphi_{n}(V_{a,b})=\true$, as required. The other cases are analogous.
This shows (\dag).

Now we can proceed to show that all individuals of $\Inter$ are contained in the extension of each concept expression of $\FLAT(\mathscr{T})$. To this end, we first show that $\selem\in\conC^\Inter$ iff $\conC\in\val(\selem)$ for all $\conC\in\parts(\FLAT(\mathscr{T}))$. If $\conC\in\connames$ is atomic, this follows directly from the definition of $\Inter$. The remaining cases that may occur in $\parts(\FLAT(\mathscr{T}))$ are $\conC=\exists\rolU.\conA$ and $\conC=\forall\rolU.\conA$.

First consider the case $\conC=\exists\rolU.\conA$ and assume that $\selem\in\conC^\Inter$. Thus there is $\selem'\in\Delta^\Inter$ with $\tuple{\selem,\selem'}\in\rolU^{\Inter}$ and $\selem'\in\conA^\Inter$. The construction of the domino model admits three possible cases:
\begin{iteMize}{$\bullet$}
\item $\selem,\selem'\in \indnames$ and $\roleval{\mathscr{R}_{\selem\selem'}}{\rolU}$ and $\conA\in\val(\selem')$. Now by (\dag), the domino $\tuple{\val(\selem),\mathscr{R}_{\selem\selem'},\val(\selem')}$ satisfies condition \textbf{ex} of Algorithm~\ref{def:builddoms}, and thus $\conC\in\val(\selem)$ as required.
\item $\selem'=\selem\tuple{\val(\selem),\mathscr{R},\val(\selem')}$ with $\roleval{\mathscr{R}}{\rolU}$ and $\conA\in\val(\selem')$. Since $\DO_{\mathscr{T}}\subseteq\DO_0$, we find that $\tuple{\val(\selem),\mathscr{R},\val(\selem')}$ satisfies condition \textbf{ex}, and thus $\conC\in\val(\selem)$ as required.
\item $\selem=\selem'\tuple{\val(\selem'),\mathscr{R},\val(\selem)}$ with $\roleval{\Inv(\mathscr{R})}{\rolU}$ and $\conA\in\val(\selem')$. By condition \textbf{sym}, $\DO_{\mathscr{T}}$ contains the domino $\tuple{\val(\selem),\Inv(\mathscr{R}),\val(\selem')}$, and again we use \textbf{ex} to conclude $\conC\in\val(\selem)$.
\end{iteMize}
For the converse, assume that $\exists\rolU.\conA \in \val(\selem)$. So $\DO_{\mathscr{T}}$ contains a domino $\tuple{\mathscr{A},\mathscr{R},\val(\selem)}$. This is obvious if the sequence $\selem$ ends with a domino. If $\selem=a\in\indnames$, then it follows by applying (\dag) to $a$ with the first individual being arbitrary. By \textbf{sym} $\DO_{\mathscr{T}}$ also contains the domino $\tuple{\val(\selem),\mathscr{R},\mathscr{A}}$. By condition \textbf{delex}, the latter implies that $\DO_{\mathscr{T}}$ contains a domino $\tuple{\val(\selem),\mathscr{R}',\mathscr{A}'}$ such that $\roleval{\mathscr{R}'}{\rolU}$ and $\conA\in\mathscr{A}'$. Thus $\selem' = \selem\tuple{\val(\selem),\mathscr{R}',\mathscr{A}'}$ is an $\Inter$-individual such that $\tuple{\selem,\selem'}\in\rolU^{\Inter}$ and $\selem'\in\conA^\Inter$, and we obtain $\selem\in (\exists\rolU.\conA)^{\Inter}$ as claimed.

For the second case, consider $\conC=\forall\rolU.\conA$ and assume that $\selem\in\conC^\Inter$. As above, we find that $\DO_{\mathscr{T}}$ contains some domino $\tuple{\mathscr{A},\mathscr{R},\val(\selem)}$, where (\dag) is needed if $\selem\in\indnames$. By \textbf{sym} we find a domino $\tuple{\val(\selem),\mathscr{R},\mathscr{A}}$. For a contradiction, suppose that $\forall\rolU.\conA\not\in\val(\selem)$. By condition \textbf{deluni}, the latter implies that $\DO_{\mathscr{T}}$ contains a domino $\tuple{\val(\selem),\mathscr{R}',\mathscr{A}'}$ such that $\roleval{\mathscr{R}'}{\rolU}$ and $\conA\notin\mathscr{A}'$. Thus $\selem' = \selem\tuple{\val(\selem),\mathscr{R}',\mathscr{A}'}$ is an $\Inter$-individual such that $\tuple{\selem,\selem'}\in\rolU^{\Inter}$ and $\selem'\notin\conA^\Inter$. But then $\selem\notin (\forall\rolU.\conA)^{\Inter}$, which is the required contradiction.

For the other direction, assume that $\forall\rolU.\conA \in \val(\selem)$. According to the construction of $\Inter$, for all elements $\selem'$ with $\tuple{\selem,\selem'}\in\rolU^{\Inter}$, there are three possible cases:
\begin{iteMize}{$\bullet$}
\item $\selem,\selem'\in \indnames$ and $\roleval{\mathscr{R}_{\selem\selem'}}{\rolU}$. Now by (\dag), the domino $\tuple{\val(\selem),\mathscr{R}_{\selem\selem'},\val(\selem')}$ satisfies condition \textbf{uni}, whence $\conA\in\val(\selem')$.
\item $\selem'=\selem\tuple{\val(\selem),\mathscr{R},\val(\selem')}$ with $\roleval{\mathscr{R}}{\rolU}$. Since $\DO_{\mathscr{T}}\subseteq\DO_0$, $\tuple{\val(\selem),\mathscr{R},\val(\selem')}$ must satisfy condition \textbf{uni}, and thus $\conA\in\val(\selem')$.
\item $\selem=\selem'\tuple{\val(\selem'),\mathscr{R},\val(\selem)}$ with $\roleval{\Inv(\mathscr{R})}{\rolU}$. By condition \textbf{sym}, $\DO_{\mathscr{T}}$ also contains the domino $\tuple{\val(\selem),\Inv(\mathscr{R}),\val(\selem')}$, and we can again use \textbf{uni} to conclude $\conA\in\val(\selem')$.
\end{iteMize}
Thus, $\conA\in\val(\selem')$ for all $\rolU$-successors $\selem'$ of $\selem$, and hence $\selem\in(\forall\rolU.\conA)^{\Inter}$ as claimed.

To finish the proof, note that any domino $\tuple{\mathscr{A},\mathscr{R},\mathscr{B}}\in\DO_{\mathscr{T}}$ satisfies condition \textbf{kb}. Using \textbf{sym}, we have that for any $\selem\in\Delta^\Inter$, the axiom $\bigsqcap_{\conD\in\val(\selem)}\conD\ssb\conC$ is a tautology for all $\conC\in\FLAT(\mathscr{T})$. As shown above, $\selem\in\conD^{\Inter}$ for all $\conD\in\val(\selem)$, and thus $\selem\in\conC^\Inter$. 
Hence every individual of $\Inter$ is an instance of each concept of $\FLAT(\mathscr{T})$ as required.
\qed

Lemmas~\ref{lemma:ddsound} and \ref{lemma:ddcomp} give rise to the following theorem which finishes the technical development of this section by showing that $\Prog(\kbrb)$ faithfully captures both positive and negative ground conclusions of $\kb$, and in particular that $\Prog(\kbrb)$ and $\kb$ are equisatisfiable.

\begin{theorem}\label{theo:aboxalcib}
For every extended \ALCIb{} knowledge base \kb{} hold
\begin{iteMize}{$\bullet$}\item
$\kb$ and $\Prog(\kb)$ are equisatisfiable,
\item
$\kb\models \conC(a)$ iff $\Prog(\kb)\models S_{\!\conC}(a)$,
\item
$\kb\models \rolR(a,b)$ iff $\Prog(\kb)\models S_{\!\rolR}(a,b)$, and
\item
$\kb\models a\approx b$ iff $\Prog(\kb)\models a\approx b$,
\end{iteMize}
for any $a,b\in\indnames$, $\conC\in\connames$, and $\rolR\in\rolnames$.
\end{theorem}
\proof
Immediate from Lemma~\ref{lemma:ddsound} and Lemma~\ref{lemma:ddcomp}.
\qed

Coming back to our example knowledge base $\mathscr{KB}$ from Section~\ref{sec:boolfunc}, the corresponding disjunctive Datalog program $\Prog(\mathscr{KB})$ contains 70 rules: two rules for each of the 33 labeled nodes from the OBDD displayed in Fig.~\ref{fig:obdd3}, the two rules $\to A_{\nroot}(x,y)$ and $A_{\nfalse}(x,y)\to$ as well as the two rules $\to S\!_{\mathit{Diploma}}(laureus)$ and $\to S\!_{\mathit{PhDStudent}}(laureus)$ introduced by conceiving the two ABox statements as DL-safe rules and translating them accordingly. The program turns out to be unsatisfiable, witnessed by the unsatisfiable subprogram displayed in Fig.~\ref{fig:subprog}.

\begin{figure}[h]
$$\begin{array}{rl@{\quad}rl}
\to & S\!_{\mathit{Diploma}}(laureus) &
\to & S\!_{\mathit{PhDStudent}}(laureus) \\
\to & A_0(x,y) & & \\
A_0(x,y) \wedge S_{\!\exists has.Diploma}(x) \to & A_5(x,y) &
A_0(x,y) \to & A_1(x,y) \vee  S_{\!\exists has.Diploma}(x) \\
A_1(x,y) \wedge S_{\!PhDStudent}(x) \to & A_{\mathsf{false}}(x,y) & & \\
A_5(x,y) \wedge S_{\!\forall has^-.Graduate}(y) \to & A_9(x,y) &
A_5(x,y) \to & A_8(x,y) \vee S_{\!\forall has^-.Graduate}(y) \\
A_8(x,y) \wedge S_{\!Graduate}(y) \to & A_{13}(x,y) &
A_8(x,y) \to & A_{12}(x,y) \vee S_{\!Graduate}(y) \\
A_9(x,y) \wedge S_{\!Graduate}(y) \to & A_{13}(x,y) &
A_9(x,y) \to & A_{16}(x,y) \vee S_{\!Graduate}(y) \\
A_{12}(x,y) \wedge S_{\!Diploma}(y) \to & A_{\mathsf{false}}(x,y) & & \\
A_{13}(x,y) \wedge S_{\!Diploma}(y) \to & A_{\mathsf{false}}(x,y) & & \\
A_{16}(x,y) \wedge S_{\!\exists has.Diploma}(y) \to & A_{\mathsf{false}}(x,y) &
A_{16}(x,y) \to & A_{20}(x,y) \vee  S_{\!\exists has.Diploma}(y) \\
A_{20}(x,y) \wedge S_{\!PhDStudent}(y) \to & A_{\mathsf{false}}(x,y) & & \\
A_{\mathsf{false}}(x,y)\to & & & \\
\end{array}$$
\caption{Unsatisfiable subprogram of $\Prog(\mathscr{KB})$ witnessing unsatisfiability of $\mathscr{KB}$}\label{fig:subprog}
\end{figure}

\section{Polynomial Transformation from \SHIQb{} to \ALCIb{}}
\label{sec:reduction}

In this section, we present a stepwise satisfiability-preserving transformation from the description logic \SHIQb{} to the more restricted \ALCIb{}. This transformation is necessary as our type-elimination method applies directly only to the latter.

\subsection{Unravelings}

For our further considerations, we will use a well-known model transformation technique which will come handy for showing equisatisfiability of knowledge base transformations introduced later on (for an introductory account on unravelings in a DL setting cf., e.g., \cite{DBLP:conf/rweb/Rudolph11}). Essentially, the transformation takes an arbitrary model of a \SHIQb{} knowledge base and converts it into a model that is ``tree-like''. We start with some preliminary definitions. The first one exploits that role subsumption on non-simple roles can be decided by an easy syntactic check that takes only role hierarchy axioms into account.


\begin{definition}
Based on a fixed \SHIQb{} knowledge base $\kb$, we define $\ssb^*$ as the smallest binary relation on the non-simple atomic roles $\Rlang_\mathrm{n}$ such that:
\begin{iteMize}{$\bullet$}
\item $\rolR\ssb^*\rolR$ for every atomic role $\rolR$,
\item $\rolR\ssb^*\rolS$ and $\Inv(\rolR)\ssb^*\Inv(\rolS)$ for every RBox axiom $\rolR\ssb\rolS$, and
\item $\rolR\ssb^*\rolT$ whenever $\rolR\ssb^*\rolS$ and $\rolS\ssb^*\rolT$ for some atomic role $\rolS$.
\end{iteMize}
Furthermore, we write $\rolR\sqsubset^*\rolS$ whenever $\rolR\ssb^*\rolS$ and $\rolS\not\ssb^*\rolR$.
\end{definition}

The next definition introduces a standard model transformation technique that is often used to show variants of the tree model property of a logic. We adopt the definition of \cite{DBLP:conf/ijcai/GlimmHLS07}.

\begin{definition}\label{def:unraveling}
Let \kb{} be a consistent extended \SHIQb{} knowledge base, and let $\Inter = \tuple{\Delta^\Inter, \cdot^\Inter}$ be a model for \kb.

The \emph{unraveling} of $\Inter$ is an interpretation that is obtained from $\Inter$ as follows. We define the set $S \subseteq (\Delta^\Inter)^{\ast}$ of \emph{sequences} to be the smallest set such that
  \begin{iteMize}{$\bullet$}
  \item for every $a \in \indnames$, $a^\Inter$ is a sequence;
  \item $\delem_1 \cdots \delem_n \cdot \delem_{n+1}$ is a sequence, if
    \begin{iteMize}{$-$}
    \item $\delem_1 \cdots \delem_n$ is a sequence,
    \item\label{paths:inverses} $\delem_{i+1} \neq \delem_{i-1}$ for all $i=2,\ldots,n$,
    \item $\tuple{\delem_n,\delem_{n+1}}\in \rolR^\Inter$ for some $\rolR \in \rolnames$.
      \end{iteMize}
  \end{iteMize}
For each $\selem = \delem_1 \cdots \delem_n \in S$, set $\tail{\selem} \defeq \delem_n$. Now, we define the unraveling of $\Inter$ as the interpretation $\Jnter = \tuple{\Delta^\Jnter, \cdot^\Jnter}$ with $\Delta^\Jnter = S$ and we define the interpretation of concept and role names as follows (where $\selem,\selem' \in \Delta^\Jnter$ are arbitrary sequences in $\Delta^\Jnter$):
  \begin{enumerate}[(a)]
  \item \label{it:nom} for each $a \in \indnames$, set $a^\Jnter \defeq a^\Inter$;
  \item \label{it:C} for each concept name $A\in\connames$, set $\selem\in A^\Jnter$ iff $\tail{\selem} \in A^\Inter$;
  \item \label{it:r} for each role name $\rolR \in \rolnames$, set $\tuple{\selem,\selem'} \in \rolR^\Jnter$ iff
  \begin{iteMize}{$\bullet$}
  \item
  $\selem' = \selem\delem$ for some $\delem\in\Delta^\Inter$ and $\tuple{\tail{\selem}, \tail{\selem'}} \in \rolR^\Inter$ or
  \item
  $\selem = \selem'\delem$ for some $\delem\in\Delta^\Inter$ and $\tuple{\tail{\selem}, \tail{\selem'}} \in \rolR^\Inter$ or
  \item
  $\selem = a^\Inter$, $\selem' = b^\Inter$ for some $a,b\in\indnames$ and $\tuple{a^\Inter, b^\Inter} \in \rolR^\Inter$.
  \end{iteMize}
  \end{enumerate}
\end{definition}

Unraveling a model of an extended \SHIQb{} knowledge base results in an interpretation that still satisfies most of the knowledge base's axioms, except for transitivity axioms. The following definition provides a ``repair strategy'' for unravelings such that also the transitivity conditions are again satisfied. The presented definition is inspired by a similar one by \cite{Motik:PhD}.

\begin{definition}\label{def:completion} 
Given an interpretation $\Inter$ and a knowledge base $\kb$, we define the \define{completion} of $\Inter$ with respect to $\kb$ as the new interpretation $\Jnter = \tuple{\Delta^\Jnter,\cdot^\Jnter}$ as follows:
\begin{iteMize}{$\bullet$}
\item
$\Delta^\Jnter\defeq \Delta^\Inter$,
\item
$a^\Jnter \defeq a^\Inter$ for every $a\in \indnames$,
\item
$\conA^\Jnter \defeq \conA^\Inter$ for every $\conA\in \connames$,
\item
for all simple roles $\rolR$, we set $\rolR^\Jnter \defeq \rolR^\Inter$,
\item
for all non-simple roles $\rolR$, $\rolR^\Jnter$ is set to the transitive closure of $\rolR^\Inter$ if $\Trans(\rolR)\in\kb$, otherwise $\rolR^\Jnter\defeq \rolR^\Inter\cup\bigcup_{\rolS\sqsubset^*\rolR\mbox{ with }\Trans(\rolS)\in\kb \mbox{ or }\Trans(\Inv(\rolS))\in\kb}(\rolS^\Inter)^*$, where $(\rolS^\Inter)^*$ denotes the transitive closure of $\rolS^\Inter$.
\end{iteMize}
\end{definition}

Having the above tools at hand, we are now ready to show that unraveling and subsequently completing a model of an extended knowledge base will result in a model. This correspondence will be helpful for showing the completeness of the knowledge base transformation steps introduced below.


\begin{lemma}\label{lemma:unrav}
Let $\kbrb$ be an extended \SHIQb{} knowledge base and let $\Inter$ be a model of $\kbrb$. Moreover, let $\Jnter$ be the unraveling of $\Inter$ and let $\Knter$ be the completion of $\Jnter$. Then the following hold:
\begin{enumerate}[\em(1)]
\item
$\Jnter$ satisfies all axioms of $\kbrb$ that are not transitivity axioms.
\item
For all sequences $\selem_1,\selem_2,\ldots,\selem_{n-1},\selem_n$ with $n>3$ and $\tuple{\selem_i,\selem_{i+1}}\in \rolR^\Jnter$ for $1\leq i\leq n$, and where $\selem_1,\selem_n\in \{a^\Jnter\mid a\in\indnames\}$ and $\selem_2,\ldots,\selem_{n-1}\not\in \{a^\Jnter\mid a\in\indnames\}$, we have $\selem_1 = \selem_n$ and $\selem_2=\selem_{n-1}$.
\item
$\Knter$ is a model of $\kbrb$.
\end{enumerate}
\end{lemma}
\proof
For the first claim, we investigate all the possible axiom types. First, as $\Inter$ and $\Jnter$ coincide w.r.t.\ concept and role memberships of all named individuals (i.e., individuals $\selem$ for which $\selem=a^\Inter$ for some $a\in \indnames$), they satisfy the same DL-safe rules.

For role hierarchy axioms $\rolU\sqsubseteq\rolV$ with $\rolU, \rolV$ restricted, suppose for a contradiction that $\Jnter$ does not satisfy $\rolU\sqsubseteq\rolV$, i.e., that there are two elements $\selem,\selem'\in \Delta^\Jnter$ such that $\tuple{\selem,\selem'}\in \rolU^\Jnter$ but $\tuple{\selem,\selem'}\not\in \rolV^\Jnter$. As $\rolU$ is restricted, either both $\selem$ and $\selem'$ are named individuals or $\selem'=\selem\delem$ or $\selem=\selem'\delem$. Therefore we know that $\tuple{\tail{\selem},\tail{\selem'}}\in \rolU^\Inter$ but $\tuple{\tail{\selem},\tail{\selem'}}\notin \rolV^\Inter$ which would violate $\rolU\sqsubseteq\rolV$ and hence, gives a contradiction.

Next, we consider TBox axioms (remember that we assume them to be normalized into axioms $\top\sqsubseteq C$ with $C$ in negation normal form). By induction on the role depth, we will show that for every concept $D$ it holds that $\selem\in D^\Jnter$ iff $\tail{\selem}\in D^\Inter$. The satisfaction of $\top\sqsubseteq C$ in $\Jnter$ then directly follows via $\Delta^\Jnter = \{\selem\in\Delta^\Jnter \mid \tail{\selem}\in \Delta^\Inter\} = \{\selem\in\Delta^\Jnter \mid \tail{\selem}\in C^\Inter\} = C^\Jnter$.

As base case, note that for $D\in \connames$, the claim follows by definition, while for $D=\top$ and $D=\bot$ the claim trivially holds. For the induction steps, note that (i) the claimed correspondence transfers immediately from concepts to their Boolean combinations and (ii) that for every $\selem\in\Delta^\Jnter$, the function $\tail{\cdot}$ gives rise to an isomorphism $\varphi$ between the neighborhood of $\selem$ in $\Jnter$ and the neighborhood of $\tail{\selem}$ in $\Inter$.
More precisely, $\varphi$ maps
$\{\selem'\in\Delta^\Jnter \mid \tuple{\selem,\selem'}\in \rolR^\Jnter \text{ for some }\rolR\in\lang{R}\}$ to $\{\delem'\in\Delta^\Inter \mid \tuple{\tail{\selem},\delem'}\in \rolR^\Inter \mbox{ for some } \rolR\in\lang{R}\}$ such that $\tuple{\selem,\selem'}\in \rolS^\Jnter$ iff $\tuple{\tail{\selem},\varphi(\selem')}\in \rolS^\Inter$ for all roles $\rolS\in\rolnames$ as well as $\selem'\in E^\Jnter$ iff $\varphi(\selem')\in E^\Inter$ for concepts $E$ that have a smaller role depth than $D$ (by induction hypothesis). Thereby, the claimed correspondence transfers to existential, universal, and cardinality restrictions as well.


For the second claim, we observe that by the definition of the unraveling, no individual $\selem=\delem_1\ldots\delem_k$ can be directly connected by some role to an individual $\selem'=\delem'_1\ldots\delem'_l$ with $\delem_1\neq\delem'_1$ unless $k=l=1$ in which case both individuals would be named by construction. On the other hand, every role chain starting from some named individual $\delem$ and not containing any other named individual contains only individuals of the form $\delem w$ with $w\in (\Delta^\Inter)^*$. Thus, we conclude that $\selem_1 = \selem_n$. Now, suppose $\selem_2\neq\selem_{n-1}$. By construction we have $\selem_2 = \selem_1 \delem$ and $\selem_{n-1} = \selem_n \delem' = \selem_1 \delem'$ with $\delem \neq \delem'$. However, then by construction, every role path from $\selem_2$ to $\selem_{n-1}$ must contain $\selem_1$ which is named and hence contradicts the assumption. Therefore $\selem_2 = \selem_{n-1}$.


Considering the third claim, we easily find that all transitivity axioms as well as role hierarchy statements are satisfied by construction.
For the TBox axioms, the argumentation is similar to the one used to prove the first claim but it has to be extended by the following observation: By construction, for all new role instances $\tuple{\selem,\selem'}\in \rolR^\Knter \setminus \rolR^\Jnter$ introduced by the completion, there is already a $\selem^*$ with $\tuple{\selem,\selem^*}\in \rolR^\Jnter$ such that $\tuple{\selem,\selem^*}\in \rolS^\Jnter$ iff $\tuple{\selem,\selem'}\in \rolS^\Inter$ for all roles $\rolS\in\rolnames$ as well as $\selem^*\in E^\Jnter$ iff $\selem'\in E^\Inter$ for concepts $E$. Therefore (and since non-simple roles are forbidden in cardinality constraints) the concept extensions do not change in $\Knter$ compared to $\Jnter$. Finally, the DL-safe rules are valid: Due to the first claim they hold in $\Jnter$. Then, they also hold in $\Knter$ since, by construction $\Knter$ and $\Jnter$ coincide when restricted to named individuals. In order to see the latter, note that $\Jnter$ also coincides with $\Inter$ w.r.t. named individuals and $\Inter$ satisfies all transitivity axioms, thus the completion does not introduce new role instances, as far as named individuals are concerned.
\qed

\subsection{From \SHIQb{} to \ALCHIQb{}}

As observed by \cite{RKH:Jelia-08}, a slight generalization of results by \cite{Motik:PhD} yields that any \SHIQbs{} knowledge base $\kb$ can be transformed into an equisatisfiable \ALCHIQb{} knowledge base. For the case of extended knowledge bases, this transformation has to be adapted in order to correctly treat the entailment of ground facts $\rolR(a,b)$ for non-simple roles $\rolR$ via transitivity. We start by defining this modified transformation, whereby the ground fact entailment is taken care of by appropriate DL-safe rules.

\newcommand{\clos}{cl}

\begin{definition}\label{def:es-trafo}
Let $\clos(\kb)$ denote the smallest set of concept expressions where
\begin{iteMize}{$\bullet$}
\item $\NNF(\neg \conC \sqcup \conD) \in \clos(\kb)$ for any TBox axiom $\conC \ssb \conD$,
\item $\conD \in \clos(\kb)$ for every subexpression $\conD$ of some concept $\conC \in \clos(\kb)$,
\item $\NNF(\neg\conC) \in \clos(\kb)$ for any $\atmost{n}\rolR.\conC \in \clos(\kb)$,
\item $\forall\rolS.\conC \in \clos(\kb)$ whenever $\Trans(\rolS) \in \kb$ and $\rolS\ssb^*\rolR$ for a role $\rolR$ with $\forall \rolR.\conC \in \clos(\kb)$.
\end{iteMize}
Finally, let $\Es(\kb)$ denote the extended knowledge base obtained from $\kb$ by removing all transitivity axioms $\Trans(\rolR)$ and
\begin{iteMize}{$\bullet$}
\item
adding the axiom $\forall \rolR.\conC \ssb \forall \rolR.(\forall \rolR.\conC)$ to $\kb$ whenever $\forall \rolR.\conC \in \clos(\kb)$,
\item
adding the axiom $\exists (\rolR \sqcap \rolR^-).\top \ssb \textit{Self}_R$ to $\kb$, where $\textit{Self}_R$ is a fresh concept,
\item
adding the DL-safe rules $\textit{Self}_R(x)\to R(x,x)$ and $R(x,y),R(y,z)\to R(x,z)$ to $\kb$.
\end{iteMize}
\end{definition}

Note that the knowledge base translation defined by $\Es$ can be done in polynomial time.
We now show that the defined transformation works as expected, making use of the model transformation techniques established in the previous section. Parts of the proof are adopted from \cite{Motik:PhD}.

\begin{proposition}\label{prop:boxpushing}
Let $\kb$ be an extended \SHIQb{} knowledge base. Then $\kb$ and $\Es(\kb)$ are equisatisfiable.
\end{proposition}

\proof Obviously, every model $\Inter$ of $\kbrb$ is a model of $\Es(\kbrb)$ if we additionally stipulate $\textit{Self}_R \defeq \{\delta \mid \tuple{\delta,\delta}\in R^\Inter\}$.

For the other direction, let $\Knter$ be a model of $\Es(\kb)$. Let now $\Inter$ be the unraveling of $\Knter$ and let $\Jnter$ be the completion of $\Inter$ w.r.t.\ $\kb$. As $\Es(\kb)$ does not contain any transitivity statements, we know by Lemma~\ref{lemma:unrav} (1) that $\Inter$ is a model of $\Es(\kb)$ as well.

As a direct consequence of the definition of the completion, note that for all simple roles $\rolV$ we have $\rolV^\Jnter = \rolV^\Inter$ (fact \dag).

We now prove that $\Jnter$ is a model of $\kb$ by considering all axioms, starting with the RBox. Every transitivity axiom of $\kb$ is obviously satisfied by the definition of $\Jnter$. Moreover, every role inclusion $\rolV\ssb \rolW$ axiom is also satisfied:

If both $\rolV$ and $\rolW$ are Boolean role expressions (which by definition contain only simple roles) this is a trivial consequence of (\dag). If $\rolV$ is a Boolean role expression and $\rolW$ is a non-simple role, this follows from (\dag) and the fact that, by construction of $\Jnter$, we have $\rolR^\Inter \subseteq \rolR^\Jnter$ for every non-simple role $\rolR$.
As a remaining case, assume that both $\rolV$ and $\rolW$ are non-simple roles. If $\rolW$ is not transitive, this follows directly from the definition, otherwise we can conclude it from the fact that the transitive closure is a monotone operation w.r.t.~set inclusion.
\medskip

We proceed by examining the concept expressions $\conC \in \clos(\kb)$ and show via structural induction that $\conC^\Inter \subseteq \conC^\Jnter$. As base case, for every concept of the form $\conA$ or $\neg\conA$ for $\conA\in \connames$ this claim follows directly from the definition of $\Jnter$. We proceed with the induction steps for all possible forms of a complex concept $\conC$ (mark that all $\conC\in \clos(\kb)$ are in negation normal form):

\begin{iteMize}{$\bullet$}
\item
Clearly, if $\conD_1^\Inter \subseteq \conD_1^\Jnter$ and $\conD_2^\Inter \subseteq \conD_2^\Jnter$ by induction hypothesis, we can directly conclude $(\conD_1\sqcap\conD_2)^\Inter\subseteq(\conD_1\sqcap\conD_2)^\Jnter$ as well as $(\conD_1\sqcup\conD_2)^\Inter\subseteq(\conD_1\sqcup\conD_2)^\Jnter$.
\item
Likewise, as we have $\rolV^\Inter\subseteq\rolV^\Jnter$ for all simple role expressions and non-simple roles $\rolV$ and again $\conD^\Inter \subseteq \conD^\Jnter$ due to the induction hypothesis, we can conclude $(\exists\rolV.\conD)^\Inter\subseteq (\exists\rolV.\conD)^\Jnter$ as well as $(\atleast{n}\rolV.\conD)^\Inter\subseteq (\atleast{n}\rolV.\conD)^\Jnter$.
\item
Now, consider $\conC = \forall\rolV.\conD$. If $\rolV$ is a simple role expression, we know that $\rolV^\Jnter = \rolV^\Inter$, whence we can derive $(\forall\rolV.\conD)^\Inter\subseteq (\forall\rolV.\conD)^\Jnter$ from the induction hypothesis.

It remains to consider the case $\conC = \forall\rolR.\conD$ for non-simple roles $\rolR$. Assume $\selem \in (\forall\rolR.\conD)^\Inter$. If there is no $\selem'$ with $\tuple{\selem,\selem'} \in \rolR^\Jnter$, then $\selem \in (\forall\rolR.\conD)^\Jnter$ is trivially true. Now assume there are such $\selem'$. For each of them, we can distinguish two cases:
\begin{iteMize}{$-$}
\item $\tuple{\selem,\selem'} \in \rolR^\Inter$, implying $\selem'\in \conD^\Inter$ and, via the induction hypothesis, $\selem'\in \conD^\Jnter$,
\item $\tuple{\selem,\selem'} \not\in \rolR^\Inter$. Yet, by construction of $\Jnter$, this means that there is a role $\rolS$ with $\rolS\ssb^*\rolR$
and $\Trans(S)\in \kb$ and a sequence $\selem=\selem_0,\ldots,\selem_n=\selem'$ with $\tuple{\selem_k,\selem_{k+1}}\in \rolS^\Inter$ for all $0\leq k < n$. Then $\selem \in (\forall\rolR.\conD)^\Inter$ implies $\selem \in (\forall\rolS.\conD)^\Inter$, and hence $\selem_1\in\conD^\Inter$. By Definition~\ref{def:es-trafo}, $\Es(\kb)$ contains the axiom $\forall \rolS.\conD \ssb \forall\rolS.(\forall\rolS.\conD)$, and hence $\selem_1\in(\forall\rolS.\conD)^\Inter$. Continuing this simple induction, we find that $\selem_k\in\conD^\Inter$ for all $k=1,\ldots,n$ including $\selem_n=\selem'$.
\end{iteMize}
So we can conclude that for all such $\selem'$ we have $\selem' \in \conD^\Inter$. Via the induction hypothesis follows $\selem \in \conD^\Jnter$ and hence we can conclude $\selem \in (\forall\rolR.\conD)^\Jnter$.
\item
Finally, consider $\conC = \atmost{n}\rolR.\conD$ and assume $\selem \in (\atmost{n}\rolR.\conD)^\Inter$. From the fact that $\rolR$ must be simple follows $\rolR^\Jnter=\rolR^\Inter$. Moreover, since both $\conD$ and $\NNF(\neg\conD)$ are contained in $\clos(\kb)$ the induction hypothesis gives $\conD^\Jnter = \conD^\Inter$. Those two facts together imply $\selem \in (\atmost{n}\rolR.\conD)^\Inter$.
\end{iteMize}

Now considering an arbitrary $\kb$ TBox axiom $\conC\ssb\conD$, we find $\NNF(\neg \conC \sqcup \conD)^\Inter=\Delta^\Inter$ as $\Inter$ is a model of $\kb$. Moreover -- by the correspondence just shown -- we have $\NNF(\neg \conC\sqcup \conD)^\Inter \subseteq \NNF(\neg \conC \sqcup \conD)^\Jnter$ and hence also $\NNF(\neg \conC \sqcup \conD)^\Jnter=\Delta^\Jnter$ making $\conC\ssb\conD$ an axiom satisfied in $\Jnter$.

For showing that all DL-safe rules from $\kb$ are satisfied, we will prove that $\Inter$ and $\Jnter$ coincide on the satisfaction of all ground atoms -- satisfaction of $\kb$ in $\Jnter$ then follows from satisfaction of $\kb$ in $\Inter$. By construction, this is obviously the case for all atoms of the shape $a\approx b$, $C(a)$ and $R(a,b)$ for $a,b\in \indnames$, $C \in \connames$ and $R\in\rolnames$ simple. Moreover we have that $\Jnter\models R(a,b)$ whenever $\Inter\models R(a,b)$. To settle the other direction, suppose $R$ non-simple and $\Jnter\models R(a,b)$ but $\Inter\not\models R(a,b)$. But then, there must be a role $S\ssb^* R$ that is declared transitive and satisfies $\Jnter\models S(a,b)$ but $\Inter\not\models S(a,b)$. Let us assume that $S$ is a minimal such role w.r.t.\ $\ssb^*$. Then, by construction, there must be a sequence $a^\Inter=\selem_1,\selem_2,\ldots,\selem_{k-1},\selem_k=b^\Inter$ with $\tuple{\selem_i,\selem_{i+1}}\in \rolS^\Inter$. This sequence can be split into subsequences at elements $o^\Inter_i$ for which there is a $o_i\in \indnames$, i.e., at named individuals, leaving us with subsequences (i) of subsequent named individuals $o_i^\Inter,o_{i+1}^\Inter$ or (ii) of the shape $o_i^\Inter=\selem_{i,1},\selem_{i,2},\ldots,\selem_{i,n-1},\selem_{i,n}=o_{i+1}^\Inter$ with $\selem_{i,2},\ldots,\selem_{i,n-1}$ unnamed individuals. For case (ii), Lemma~\ref{lemma:unrav} (2) guarantees $o_i^\Inter = o_{i+1}^\Inter$ and $\selem_{i,2}=\selem_{i,n-1}$, which implies $o_i^\Inter \in (\exists(\rolR\sqcap\rolR^-).\top)^\Inter$. Then, due to the according axiom $\exists (\rolR \sqcap \rolR^-).\top \ssb \textit{Self}_R$ in $\Es(\kb)$, we obtain $o_i^\Inter \in \textit{Self}_R^{\,\Inter}$ and by the DL-safe rule $\textit{Self}_R(x)\to R(x,x)$ we have $\tuple{o_i^\Inter,o_i^\Inter}\in \rolR^\Inter$. Hence, we know that $R(o_i,o_{i+1})$ holds in $\Inter$ for all our subsequences $o_i^\Inter\ldots o_{i+1}^\Inter$. But then, a (possibly iterated) application of the DL-safe rule $R(x,y)\wedge R(y,z)\to R(x,z)$ also yields that $R(a,b)$ is valid in $\Inter$, contradicting our assumption. This finishes the proof.
\qed

\subsection{From \ALCHIQb{} to \ALCHIqb}

We now show how any extended \ALCHIQb{} knowledge base $\kb$ can be transformed into an extended \ALCHIqb{} knowledge base $\Ege(\kb)$. The difference between the two DLs is that the latter does not allow $\atleast{}$ number restrictions. This transformation (as well as the one presented in Section~\ref{sec:num2}) makes use of the Boolean role constructors and differs conceptually and technically from another method for removing qualified number restrictions from DLs described by \cite{DBLP:conf/aaai/GiacomoL94}.

Given an \ALCHIQb{} knowledge base $\kb$, the \ALCHIqb{} knowledge base $\Ege(\kb)$ is obtained by first flattening $\kb$ and then iteratively applying the following procedure to $\FLAT(\kb)$, terminating if no $\atleast{}$ restrictions are left:

\begin{iteMize}{$\bullet$}
\item
Choose an occurrence of $\atleast{n}\rolU.\conA$ in the knowledge base.
\item
Substitute this occurrence by $\exists \rolR_1.\conA \sqcap \ldots \sqcap \exists \rolR_n.\conA$, where $\rolR_1,\ldots,\rolR_n$ are fresh role names.
\item
For every $i\in\{1,\ldots,n\}$, add $ \rolR_i \ssb \rolU$ to the knowledge base's RBox.
\item
For every $1 \leq i < k \leq n$, add $\forall (\rolR_i \sqcap \rolR_k).\bot$ to the knowledge base.
\end{iteMize}

Observe that this transformation can be done in polynomial time, assuming a unary encoding of the numbers $n$. It remains to show that $\kb$ and $\Ege(\kb)$ are indeed equisatisfiable.

\begin{lemma}\label{lemma:removefirstq}
Let $\kb$ be an extended \ALCHIQb{} knowledge base. Then we have that the extended \ALCHIqb{} knowledge base $\Ege(\kb)$ and $\kb$ are equisatisfiable.
\end{lemma}
\proof
%
First we prove that every model of $\Ege(\kb)$ is a model of $\kb$. We do so by an inductive argument, showing that no additional models can be introduced in any substitution step of the above conversion procedure. Hence, assume $\kb''$ is an intermediate knowledge base that has a model $\Inter$, and that is obtained from $\kb'$ by eliminating the occurrence of $\atleast{n}\rolU.\conA$ as described above. Considering $\kb''$, we find due to the $\kb''$ axioms $\forall (\rolR_i \sqcap \rolR_k).\bot$ that no two individuals $\delem,\delem'\in \Delta^\Inter$ can be connected by more than one of the roles $\rolR_1,\ldots,\rolR_n$. In particular, this enforces $\delem' \neq \delem''$, whenever $\tuple{\delem,\delem'}\in \rolR_i^\Inter$ and $\tuple{\delem,\delem''}\in \rolR_j^\Inter$ for distinct $\rolR_i$ and $\rolR_j$. Now consider an arbitrary $\delem\in(\exists \rolR_1.\conA \sqcap \ldots \sqcap \exists \rolR_n.\conA)^\Inter$. This ensures the existence of individuals $\delem_1,\ldots,\delem_n$ with $\tuple{\delem,\delem_i}\in \rolR_i^\Inter$ and $\delem_i\in \conA^\Inter$ for $1\leq i \leq n$. By the above observation, all such $\delem_i$ are pairwise distinct. Moreover, the axioms $\rolR_i \ssb \rolU$ ensure $\tuple{\delem,\delem_i}\in \rolU^\Inter$ for all $i$, hence we find that $\delem \in (\atleast{n}\rolU.\conA)^\Inter$. So we know $(\exists \rolR_1.\conA \sqcap \ldots \sqcap \exists \rolR_n.\conA)^\Inter \subseteq (\atleast{n}\rolU.\conC)^\Inter$. From the fact that both of those concept expressions occur outside any negation or quantifier scope (as the transformation starts with a flattened knowledge base and does not itself introduce such nestings) in axioms $\conD''\in \kb''$ and $\conD'\in \kb'$ which are equal up to the substituted occurrence, we can derive that $\conD''^\Inter \subseteq \conD'^\Inter$. Then, from $\conD''^\Inter = \Delta^\Inter$ follows $\conD'^\Inter = \Delta^\Inter$ making $\conD'$ valid in $\Inter$. Apart from $\conD'$, all other axioms from $\kb'$ coincide with those from $\kb''$ and hence are naturally satisfied in $\Inter$. So we find that $\Inter$ is a model of $\kb'$.

At the end of our inductive chain, we finally arrive at $\FLAT(\kb)$ which is equisatisfiable to $\kb$ by Proposition~\ref{prop:flateq}.

Second, we show that $\Ege(\kb)$ has a model if $\kb$ has. By Proposition~\ref{prop:flateq}, satisfiability of $\kb$ entails the existence of a model of $\FLAT(\kb)$. Moreover, every model of $\FLAT(\kb)$ can be transformed to a model of $\Ege(\kb)$, as we will show using the same inductive strategy as above by doing iterated model transformations following the syntactic knowledge base conversions. Again, assume $\kb''$ is an intermediate knowledge base obtained from $\kb'$ by eliminating the occurrence of $\atleast{n}\rolU.\conA$ as described above, and suppose $\Inter$ is a model of $\kb'$. Based on $\Inter$, we now (nondeterministically) construct an interpretation $\Jnter$ as follows:
\begin{iteMize}{$\bullet$}
\item
$\Delta^\Jnter \defeq \Delta^\Inter$,
\item
for all $\conC\in\connames$, let $\conC^\Jnter \defeq \conC^\Inter$,
\item
for all $\rolS\in\rolnames\setminus\{\rolR_i\mid 1\leq i \leq n\}$,
 let $\rolS^\Jnter \defeq \rolS^\Inter$,
\item
 for every $\delem \in (\atleast{n}\rolU.\conA)^\Inter$, choose
 pairwise distinct $\epsilon^\delem_1, \ldots, \epsilon^\delem_n$
 with $\tuple{\delem,\epsilon^\delem_i}\in \rolU^\Inter$ and
 $\epsilon^\delem_i\in\conA^\Inter$ (their existence being ensured
 by $\delem$'s aforementioned concept membership) and let
$\rolR_i^\Jnter \defeq \{\tuple{\delem,\epsilon^\delem_i}\mid\delem \in (\atleast{n}\rolU.\conA)^\Inter \}$.
\end{iteMize}
Now, it is easy to see that $\Jnter$ satisfies all newly introduced axioms of the shape $\forall (\rolR_i \sqcap \rolR_k).\bot$, as the $\epsilon^\delem_i$ have been chosen to be distinct for every $\delem$. Moreover the axioms $\rolR_i \ssb \rolU$ are obviously satisfied by construction. Finally, for all $\delem \in (\atleast{n}\rolU.\conA)^\Inter$ the construction ensures $\delem \in (\exists \rolR_1.\conA \sqcap \ldots \sqcap \exists \rolR_n.\conA)^\Jnter$ witnessed by the respective $\epsilon^\delem_i$. So we have $(\atleast{n}\rolU.\conA)^\Inter \subseteq (\exists\rolR_1.\conA \sqcap \ldots \sqcap \exists \rolR_n.\conA)^\Jnter$. Now, again exploiting the fact that both of those concept expressions occur in negation normalized universal concept axioms $\conD'\in \kb'$ and $\conD''\in \kb''$ that are equal up to the substituted occurrence, we can derive that $\conD'^\Inter \subseteq \conD''^\Jnter$. Then, from $\conD'^\Inter = \Delta^\Inter$ follows $\conD''^\Jnter = \Delta^\Jnter$ making $\conD''$ valid in $\Jnter$. Apart from $\conD'$ (and the newly introduced axioms considered above), all other axioms from $\kb''$ coincide with those from $\kb'$ and hence are satisfied in $\Jnter$, as they do not depend on the $\rolR_i$ whose interpretations are the only ones changed in $\Jnter$ compared to $\Inter$. So we find that $\Jnter$ is a model of $\kb''$.
\qed

\subsection{From \ALCHIqb{} to \ALCIqb{}}

In the presence of restricted role expressions, role subsumption axioms can be easily transformed into TBox axioms, as the subsequent lemma shows. This allows to dispense with role hierarchies in \ALCHIqb{} thereby restricting it to \ALCIqb{}.

\begin{lemma}
For any two restricted role expressions $\rolU$ and $\rolV$, the RBox axiom $\rolU \ssb \rolV$ and the TBox axiom $\forall (U\sqcap\neg V).\bot$ are equivalent.
\end{lemma}
\proof
By the semantics' definition, $\rolU \ssb \rolV$ holds in an interpretation $\Inter$ exactly if for every two individuals $\delem,\delem'$ with $\tuple{\delem,\delem'} \in \rolU^\Inter$ it also holds that $\tuple{\delem,\delem'} \in \rolV^\Inter$. This in turn is the case if and only if there are no $\delem,\delem'$ with $\tuple{\delem,\delem'} \in \rolU^\Inter$ but $\tuple{\delem,\delem'} \not\in \rolV^\Inter$ (the latter being expressible as $\tuple{\delem,\delem'} \in (\neg\rolV)^\Inter$). This condition can be formulated as $(\rolU\sqcap\neg \rolV)^\Inter = \emptyset$, which is equivalent to $\forall (\rolU\sqcap\neg \rolV).\bot$.
\qed

Note that $\rolU\sqcap\neg \rolV$ is restricted (hence an admissible role expression) whenever $\rolU$ is -- this can be seen from the fact that $\nroleval{\emptyset}{\rolU}$ implies $\nroleval{\emptyset}{\rolU\sqcap\neg \rolV}$ due to the definition of $\vdash$ and the Boolean role operator $\sqcap$.
 Consequently, for any extended \ALCHIqb{} knowledge base $\kb$, let $\Eh(\kb)$ denote the \ALCIqb{} knowledge base obtained by substituting every RBox axiom $\rolU\ssb \rolV$ by the TBox axiom $\forall (\rolU\sqcap\neg \rolV).\bot$. The above lemma assures equivalence of $\kb$ and $\Eh(\kb)$ (and hence also their equisatisfiability). Obviously, this reduction can be done in linear time.

\subsection{From \ALCIqb{} to \ALCIFb{}}\label{sec:num2}

The elimination of the $\atmost{}$ concept descriptions from an extended \ALCIqb{} knowledge base is more intricate than the previously described transformations. Thus, to simplify our subsequent presentation, we assume that all Boolean role expressions $\rolU$ occurring in concept expressions of the shape $\atmost{n}\rolU.\conC$ are atomic, i.e. $\rolU\in \lang{R}$. This can be easily achieved by introducing a new role name $\rolR_\rolU$ and substituting $\atmost{n}\rolU.\conC$ by $\atmost{n}\rolR_\rolU.\conC$ as well as adding the two TBox axioms
$\forall(\rolU \sqcap \neg \rolR_\rolU).\bot$ and $\forall(\neg \rolU \sqcap \rolR_\rolU).\bot$ (this ensures that the interpretations of $\rolU$ and $\rolR_\rolU$ always coincide).

To further make the presentation more conceivable, we subdivide it into two steps: first we eliminate concept expressions of the shape $\atmost{n}\rolR.\conC$ merely leaving axioms of the form $\atmost{1}\rolR.\top$ (also known as role functionality statements) as the only occurrences of number restrictions, hence obtaining an \ALCIFb{} knowledge base.\footnote{Following the notational convention, we use $\mathcal{F}$ to indicate the modeling feature of role functionality.} Then, in a second step discussed in the next section, we eliminate all occurrences of axioms of the shape $\atmost{1}\rolR.\top$.

Let $\kb$ an \ALCIqb{} knowledge base. We obtain the \ALCIFb{} knowledge base $\Ele(\kb)$ by first flattening $\kb$ and then successively applying the following steps (stopping when no further such occurrence is left):
\begin{iteMize}{$\bullet$}
\item
Choose an occurrence of the shape $\atmost{n}\rolR.\conA$ which is not a functionality axiom $\atmost{1}\rolR.\top$,
\item
substitute this occurrence by $\forall (\rolR \sqcap \neg \rolR_1 \sqcap \ldots \sqcap \neg \rolR_n).\neg \conA$ where $\rolR_1, \ldots, \rolR_n$ are fresh role names,
\item
for every $i\in\{1,\ldots,n\}$, add $\forall \rolR_i.\conA$ as
well as $\atmost{1}\rolR_i.\top$ to the knowledge base.\\
\end{iteMize}

This transformation can clearly be done in polynomial time, again assuming a unary encoding of the number $n$. We now show that this conversion yields an equisatisfiable extended knowledge base. Structurally, the proof is similar to that of Lemma \ref{lemma:removefirstq}.

\begin{lemma}
Given an extended \ALCIqb{} knowledge base $\kb$, the extended \ALCIFb{} knowledge base $\Ele(\kb)$ and $\kb$ are equisatisfiable.
\end{lemma}
\proof
$\kb$ and $\FLAT(\kb)$ are equisatisfiable by Proposition \ref{prop:flateq}, so it remains to show equisatisfiability of $\FLAT(\kb)$ and $\Ele(\kb)$.

First, we prove that every model of $\Ele(\kb)$ is a model of $\FLAT(\kb)$. We do so in an inductive way by showing that no additional models can be introduced in any substitution step of the above conversion procedure. Hence, assume $\kb''$ is an intermediate knowledge base with model $\Inter$, and that is obtained from $\kb'$ by eliminating the occurrence of $\atmost{n}\rolR.\conA$ as described above. Now consider an arbitrary $\delem\in(\forall (\rolR \sqcap \neg \rolR_1 \sqcap \ldots \sqcap \neg \rolR_n).\neg \conA)^\Inter$. This ensures that whenever an individual $\delem'\in\Delta^\Inter$ satisfies $\tuple{\delem,\delem'}\in \rolR^\Inter$ and $\delem'\in \conA$, it must additionally satisfy $\tuple{\delem,\delem'}\in \rolR_i^\Inter$ for one $i\in \{1,\ldots,n\}$. However, it follows from the $\kb''$-axioms $\atmost{1}\rolR_i.\top$ that there is at most one such $\delem'$ for each $\rolR_i$. Thus, there can be at most $n$ individuals $\delem'$ with $\tuple{\delem,\delem'}\in \rolR^\Inter$ and $\delem'\in \conA$. This implies $\delem \in (\atmost{n}\rolR.\conA)^\Inter$. So we have  $(\forall (\rolR \sqcap \neg \rolR_1 \sqcap \ldots \sqcap \neg \rolR_n).\neg \conA)^\Inter \subseteq (\atmost{n}\rolR.\conA)^\Inter$. Due to the flattened knowledge base structure, both of those concept expressions occur outside the scope of any negation or quantifier within axioms $\conD''\in \kb''$ and $\conD'\in \kb'$ that are equal up to the substituted occurrence. Hence, we can derive that $\conD''^\Inter \subseteq \conD'^\Inter$. Then, from $\conD''^\Inter = \Delta^\Inter$ follows $\conD'^\Inter = \Delta^\Inter$ making $\conD'$ valid in $\Inter$. Apart from $\conD'$, all other axioms from $\kb'$ are contained in $\kb''$ and hence are naturally satisfied in $\Inter$. So we find that $\Inter$ is a model of $\kb'$ as well.

Second, we show that every model of $\FLAT(\kb)$ can be transformed to a model of $\Ele(\kb)$. We use the same induction strategy as above by doing iterated model transformations following the syntactic knowledge base conversions. Again, assume $\kb''$ is an intermediate knowledge base obtained from $\kb'$ by eliminating the occurrence of a $\atmost{n}\rolR.\conC$ as described above, and suppose $\Inter$ is a model of $\kb'$. Based on $\Inter$, we now (nondeterministically) construct an interpretation $\Jnter$ as follows:
\begin{iteMize}{$\bullet$}
\item
$\Delta^\Jnter \defeq \Delta^\Inter$,
\item
for all $\conC\in\connames$, let $\conC^\Jnter \defeq \conC^\Inter$,
\item
for all $\rolS\in\rolnames\setminus\{\rolR_i\mid 1\leq i \leq n\}$,
 let $\rolS^\Jnter \defeq \rolS^\Inter$,
\item
 for every $\delem \in (\atmost{n}\rolR.\conA)^\Inter$, let
 $\epsilon^\delem_1, \ldots, \epsilon^\delem_k$ be an exhaustive
 enumeration (with arbitrary but fixed order)
 of all those $\epsilon\in\Delta^\Inter$
 with $\tuple{\delem,\epsilon}\in \rolR^\Inter$ and
 $\epsilon\in\conA^\Inter$. Thereby $\delem$'s aforementioned
 concept membership ensures $k\leq n$. Now, let
$\rolR_i^\Jnter \defeq \{\tuple{\delem,\epsilon^\delem_i}\mid\delem \in (\atmost{n}\rolR.\conA)^\Inter \}$.
\end{iteMize}
Now, it is easy to see that $\Jnter$ satisfies all newly introduced axioms of the shape $\atmost{1}\rolR_i.\top$ as every $\delem$ has at most one $\rolR_i$-successor (namely $\epsilon^\delem_i$, if $\delem\in (\atmost{n}\rolR.\conA)^\Inter$, and none otherwise). Moreover, the axioms $\forall \rolR_i.\conA$ are satisfied, as the $\epsilon^\delem_i$ have been chosen accordingly.

Finally for all $\delem \in (\atmost{n}\rolR.\conA)^\Inter$ the construction ensures $\delem \in (\forall (\rolR \sqcap \neg\rolR_1 \sqcap \ldots \sqcap \neg \rolR_n).\neg \conA)^\Jnter$ as by construction, each $\rolR$-successor of $\delem$ that lies within the extension of $\conA$ is contained in $\epsilon^\delem_1, \ldots, \epsilon^\delem_k$ and therefore also $\rolR_i$-successor of $\delem$ for some $i$. Now, again exploiting the fact that both of those concept expressions occur in negation normalized universal concept axioms $\conD'\in \kb'$ and $\conD''\in \kb''$ that are equal up to the substituted occurrence, we can derive that $\conD'^\Inter \subseteq \conD''^\Jnter$. Then, from $\conD'^\Inter = \Delta^\Inter$ follows $\conD''^\Jnter = \Delta^\Jnter$ making $\conD''$ valid in $\Jnter$. Apart from $\conD''$ (and the newly introduced axioms considered above), all other axioms from $\kb''$ coincide with those from $\kb'$ and hence are satisfied in $\Jnter$, as they do not depend on the $\rolR_i$ whose interpretations are the only ones changed in $\Jnter$ compared to $\Inter$. So we find that $\Jnter$ is a model of $\kb''$.
\qed

\subsection{From \ALCIFb{} to \ALCIb{}}

In the sequel, we show how the role functionality axioms of the shape $\atmost{1}\rolR.\top$ can be eliminated from an \ALCIFb{} knowledge base while still preserving equisatisfiability. Partially, the employed rewriting is the same as the one proposed for $\mathcal{ALCIF}$ TBoxes by \cite{DL-98-alc}, however, in the presence of ABoxes more needs to be done.

Essentially, the idea is to add axioms that enforce that for every functional role $\rolR$, any two $\rolR$-successors coincide with respect to their properties expressible in ``relevant'' DL role and concept expressions. To this end, we consider the parts of a knowledge base as defined in Section~\ref{sec:prelims} on page \pageref{page:parts}. While it is not hard to see that the introduced axioms follow from $\rolR$'s functionality, the other direction (a Leibniz-style ``identitas indiscernibilium'' argument) needs a closer look.

Taking an extended \ALCIFb{} knowledge base $\kb$, let $\Ef(\kb)$ denote the extended \ALCIb{} knowledge base obtained from $\kb$ by removing every role functionality axiom $\atmost{1}\rolR.\top$ and instead adding

\begin{iteMize}{$\bullet$}
\item
$\forall \rolR.\neg\conD \sqcup \forall \rolR.\conD$ for every $\conD\in
\parts(\kb\setminus\{\alpha \in \kb\mid \alpha=\atmost{1}\rolR.\top\text{ for some }\rolR\in\lang{R}\})$,
\item
 $\forall (\rolR \sqcap \rolS).\bot \sqcup \forall (\rolR \sqcap \neg
 \rolS).\bot$
for every atomic role $\rolS$ from $\kb$, as well as
\item
the DL-safe rule $\rolR(x,y),\rolR(x,z) \to y\approx z$.
\end{iteMize}

Clearly, this transformation can also be done in polynomial time and space w.r.t.~the size of $\kb$.

Our goal is now to prove equisatisfiability of $\kb$ and $\Ef(\kb)$. The following lemma establishes the easier direction of this correspondence.

\begin{lemma}\label{lemma:kb4kb5}
Any \ALCIFb{} knowledge base $\kb$ entails all axioms of the \ALCIb{} knowledge base $\Ef(\kb)$, i.e. $\kb \models \Ef(\kb)$.
\end{lemma}
\proof
Let $\Jnter$ be a model of $\kb$. We need to show that $\Jnter$ also satisfies the additional rules and axioms introduced in $\Ef(\kb)$.

First let $\conD$ be an arbitrary concept. Note that $\forall \rolR.\neg\conD \sqcup \forall \rolR.\conD$ is equivalent to the GCI $\exists\rolR.\conD \ssb \forall\rolR.\conD$. This is satisfied if, for any $\delta\in \Delta^\Jnter$, if $\delta$ has an $\rolR$-successor in $\conD^\Jnter$, then all $\rolR$-successors of $\delta$ are in $\conD^\Jnter$.
This is trivially satisfied if $\delta$ has at most one $\rolR$-successor, which holds since $\Jnter$ satisfies the functionality axiom $\atmost{1}\rolR.\top\in\kb$. Since we have shown the satisfaction for arbitrary concepts $\conD$, this holds in particular for those from $\parts(\kb\setminus\{\alpha \in \kb\mid \alpha=\atmost{1}\rolR.\top\text{ for some }\rolR\in\lang{R}\})$.

Second, let $\rolS$ be an atomic role. Mark that $\forall (\rolR \sqcap \rolS).\bot \sqcup \forall (\rolR \sqcap \neg \rolS).\bot$ is equivalent to the GCI $\exists (\rolR \sqcap \rolS).\top \ssb \forall (\rolR \sqcap \neg \rolS).\bot$. This means that for any $\delta\in \Delta^\Jnter$, all $\rolR$-successors are also $\rolS$-successors of it, whenever one of them is. Again, this is trivially satisfied as $\delta$ has at most one $\rolR$-successor.

Finally all newly introduced rules of the form $R(x,y),R(x,z) \to y\approx z$ are satisfied in $\Jnter$ as a consequence of the functionality statements in \kb.
\qed

The other direction for showing equisatisfiability, which amounts to finding a model of $\kb$ given one for $\Ef(\kb)$, is somewhat more intricate and requires some intermediate considerations.

\begin{lemma}\label{lemma:identitas}
If $\kb$ is an \ALCIFb{} knowledge base with $\atmost{1}\rolR.\top \in \kb$ then in every model $\Jnter$ of $\Ef(\kb)$ we find that $\tuple{\delta,\delta_1} \in \rolR^\Jnter$ and $\tuple{\delta,\delta_2} \in \rolR^\Jnter$ imply
\begin{iteMize}{$\bullet$}
\item
for all $\conC \in \parts(\kb\setminus\{\alpha \in \kb\mid \alpha=\atmost{1}\rolR.\top\text{ for some }\rolR\in\lang{R}\})$, we have $\delta_1\in \conC^\Jnter$ iff $\delta_2\in \conC^\Jnter$,
\item
for all $\rolS\in\rolnames$, we have $\tuple{\delta,\delta_1}\in \rolS^\Jnter$ iff $\tuple{\delta,\delta_2}\in \rolS^\Jnter$.
\end{iteMize}
\end{lemma}
\proof
For the first proposition, assume $\delta_1\in \conC^\Jnter$. From $\tuple{\delta,\delta_1} \in \rolR^\Jnter$ follows $\delta \in (\exists\rolR.\conC)^\Jnter$. Due to the $\Ef(\kbrb)$ axiom $\forall \rolR.\neg\conC \sqcup \forall \rolR.\conC$ (being equivalent to the GCI $\exists\rolR.\conC \ssb \forall\rolR.\conC$) follows $\delta \in (\forall\rolR.\conC)^\Jnter$. Since $\tuple{\delta,\delta_2} \in \rolR^\Jnter$, this implies $\delta_2\in \conC^\Jnter$. The other direction follows by symmetry.

To show the second proposition, assume $\tuple{\delta,\delta_1}\in \rolS^\Jnter$. Since also $\tuple{\delta,\delta_1} \in \rolR^\Jnter$, we have $\tuple{\delta,\delta_1} \in \rolR\sqcap\rolS^\Jnter$ and hence $\delta \in (\exists (\rolR \sqcap \rolS).\top)^\Jnter$ . From the $\Ef(\kbrb)$ axiom $\forall (\rolR \sqcap \rolS).\bot \sqcup \forall (\rolR \sqcap \neg
 \rolS).\bot$ (which is equivalent to the GCI $\exists (\rolR
\sqcap \rolS).\top \ssb \neg \exists (\rolR \sqcap \neg
 \rolS).\top$) we conclude $\delta \in (\neg\exists (\rolR \sqcap \neg
 \rolS).\top)^\Jnter$, in words: $\delta$ has no $\rolR$-successor that
 is not its $\rolS$-successor. Thus, as $\tuple{\delta,\delta_2} \in
 \rolR^\Jnter$, it must also hold that $\tuple{\delta,\delta_2}\in
 \rolS^\Jnter$. Again, the other direction follows by symmetry.
\qed

In order to convert a model of $\Ef(\kbrb)$ into one of $\kbrb$, we will have to enforce role functionality where needed by cautiously deleting individuals from the original model. 
Definition~\ref{def:pruning} will provide a method for this. To this end, some auxiliary notions 
defined beforehand will come in handy.

\begin{definition} Let $\Jnter$ be an interpretation, and let $\Inter$ be the unraveling of $\Jnter$.\footnote{Remember that by construction, the individuals of $\Inter$ are sequences of individuals of $\Jnter$. For better readability, we will strictly use $\selem$ -- with possible subscripts -- for $\Inter$-individuals and $\delem$ for $\Jnter$-individuals.}
For a domain element $\selem \in \Delta^\Inter$ and an $\rolR \in \lang{R}$, we define the set of \define{$\rolR$-neighbors} of $\selem$ in $\Inter$ by $\mbox{\small\textsf{nb}}_\Inter^\rolR(\selem)\defeq \{\selem' \mid \tuple{\selem,\selem'}\in \rolR^\Inter\}$. Among the \define{$\rolR$-neighbors}, we distinguish between  \define{subordinate $\rolR$-neighbors} $\mbox{\small\textsf{sub}}_\Inter^\rolR(\selem)\defeq \{\selem\delem \mid \tuple{\selem,\selem\delem}\in \rolR^\Inter\}$ and the \define{non-subordinate $\rolR$-neighbors} $\mbox{\small\textsf{nonsub}}_\Inter^\rolR(\selem)\defeq \mbox{\small\textsf{nb}}_\Inter^\rolR(\selem)\setminus \mbox{\small\textsf{sub}}_\Inter^\rolR(\selem)$.
\end{definition}


%

\begin{definition}\label{def:pruning}
Let $\Jnter$ be an interpretation, and let $\Inter$ be the unraveling of $\Jnter$.
 Given an extended \ALCIFb{} knowledge base \kbrb, let $\kb^* \defeq \kb\setminus\{\alpha \in \kb\mid \alpha=\atmost{1}\rolR.\top\text{ for some }\rolR\in\lang{R}\}$, let $\mathscr{D} \defeq \parts(\kbrb)$ and let $\mathscr{S}\defeq\{\rolR \mid \atmost{1}\rolR.\top \in \kb\}$.

Then, an interpretation $\Knter$ will be called \define{\kb-pruning} of $\Inter$, if $\Knter$ can be constructed from $\Inter$ in the following way: Let first $\Delta_0=\Delta^\Inter$.
%
%
%
Next, iteratively determine $\Delta_{i+1}$ from $\Delta_i$ as follows:
\begin{iteMize}{$\bullet$}
\item
Select a word-length minimal $\sigma$ from $\Delta_i$ where there is an $\rolS\in\mathscr{S}$ for which  
$\mbox{\small\textsf{nb}}_\Inter^\rolS(\selem)>1$ and $\mbox{\small\textsf{sub}}_\Inter^\rolS(\selem)>0$.

\item

If $\mbox{\small\textsf{nonsub}}_\Inter^\rolS(\selem)>0$, let $\Delta'=\mbox{\small\textsf{sub}}_\Inter^\rolS(\selem)$, otherwise let $\Delta'=\mbox{\small\textsf{sub}}_\Inter^\rolS(\selem)\setminus\{\sigma'\}$ for an arbitrarily chosen $\sigma'\in\mbox{\small\textsf{sub}}_\Inter^\rolS(\selem)$.\\
Delete $\Delta'$ from $\Delta_i$ as well as all $\sigma^{**}$ having some $\sigma^*\in\Delta'$ as prefix.
\end{iteMize}
Finally, let $\Knter$ be the limit of this process: $\Delta^\Knter \defeq \bigcap_{i\in\mathbb{N}}\Delta_i$ and $\cdot^\Knter$ is the function $\cdot^\Inter$ restricted to $\Delta^\Knter$.
\end{definition}

Roughly speaking, any \define{$\kb$-pruning} of $\Inter$ is (nondeterministically) constructed by deleting surplus functional-role-successors. Mark that the tree-like structure of non-named individuals of the unraveling is crucial in order to make the process well-defined.

\begin{lemma}\label{lemma:modelforf}
Let $\kbrb$ be an extended \ALCIFb{} knowledge base, let $\Jnter$ be a model of $\Ef(\kbrb)$ and let $\Inter$ be an unraveling of $\Jnter$. Then, any $\kbrb$-pruning $\Knter$ of $\Inter$ is a model of $\kbrb$.
\end{lemma}
\proof
By construction, we know that $\Inter$ is a model of $\Ef(\kb)$. Now, let $\Knter$ be a $\kb$-pruning of $\Inter$. For showing $\Knter \models \kb$, we divide $\kb$ into two sets, namely the set of role functionality axioms $\{\alpha \in \kb\mid \alpha=\atmost{1}\rolR.\top\text{ for some }\rolR\in\lang{R}\}$ and all the remaining axioms, denoted by $\kb^*$, and show $\Knter \models \kb^*$ and $\Knter \models \{\alpha \in \kb\mid \alpha=\atmost{1}\rolR.\top\text{ for some }\rolR\in\lang{R}\}$ separately.

\medskip

We start by showing $\Knter \models \kb^*$.
To this end, we prove that, for each $\conC\in
\parts(\kb^*)$ and for every individual $\sigma$ from $\Knter$, we
have $\sigma \in \conC^\Knter$ exactly if $\sigma \in \conC^\Inter$. Clearly, this statement extends to concepts that are Boolean combinations of elements from $\parts(\kb^*)$, i.e., to all axioms in $\kb^*$. We omit this easy structural induction.

The claim for $\conC\in
\parts(\kb^*)$ is shown by induction over the depth of role restrictions in $\conC$, and we assume that is has already been shown for concepts of smaller role depth. We consider three cases:
\begin{iteMize}{$\bullet$}
\item
$\conC \in \connames\cup\{\top,\bot\}$\\ Then the coincidence follows directly from the construction of $\Knter$.

\item
$\conC = \exists \rolU. \conD$\\
``$\Rightarrow$''\mbox{}\hspace{2mm}\mbox{}%
$\sigma \in (\exists \rolU. \conD)^\Knter$ means that there is a $\Knter$-individual $\sigma'$ with $\tuple{\sigma,\sigma'}\in\rolU^{\Knter}$ and $\sigma'\in \conD^\Knter$. Because of the construction of $\Knter$ by pruning $\Inter$, this means also $\tuple{\sigma,\sigma'}\in\rolU^{\Inter}$ and by induction hypothesis, we have $\sigma'\in \conD^\Inter$, ergo $\sigma \in (\exists \rolU. \conD)^\Inter$.

\noindent
``$\Leftarrow$''\mbox{}\hspace{2mm}\mbox{}%
If $\sigma \in (\exists \rolU. \conD)^\Inter$, there is an $\Inter$-individual $\sigma'$ with $\tuple{\sigma,\sigma'}\in\rolU^{\Inter}$ and $\sigma'\in\conD^\Inter$. In case $\sigma'$ is not deleted during the construction of $\Knter$, it proves (by using the induction hypothesis on $\conD$) that $\sigma \in (\exists \rolU. \conD)^\Knter$. Otherwise, it must have been deleted due to the existence of another $\Inter$-individual $\sigma''$ for with Lemma~\ref{lemma:identitas} ensures $\{\rolR\in \lang{R}\mid \tuple{\sigma,\sigma''} \in \rolR^\Inter\}=\{\rolR\in \lang{R}\mid \tuple{\sigma,\sigma'} \in \rolR^\Inter\}$ and $\{\conE\in \parts(\kb^*)\mid \sigma''\in \conE^\Inter\}=\{\conE\in
\parts(\kb^*)\mid \sigma'\in \conE^\Inter\}$. W.l.o.g., $\sigma''$
does not get deleted in the whole construction procedure. Yet, then the $\Knter$-individual $\sigma''$ obviously proves $\sigma \in (\exists \rolU. \conD)^\Knter$.

\item
$\conC = \forall \rolR. \conD$\\
``$\Rightarrow$''\mbox{}\hspace{2mm}\mbox{}%
Assume the contrary, i.e., $\sigma \in (\forall \rolU. \conD)^\Knter$ but $\sigma \not\in (\forall \rolU. \conD)^\Inter$ which means that there is an $\Inter$-individual $\sigma'$ with $\tuple{\sigma,\sigma'}\in\rolU^{\Inter}$ but $\sigma' \not\in\conD^\Inter$. In case $\sigma'$ has not been deleted during the construction of $\Knter$, it disproves $\sigma \in (\forall \rolU. \conD)^\Knter$ (by invoking the induction hypothesis on $\conD$) leading to a contradiction. Otherwise, $\sigma'$ is deleted because of the existence of another $\Inter$-individual $\sigma''$ for with Lemma~\ref{lemma:identitas} ensures  $\{\rolR\in \lang{R}\mid \tuple{\sigma,\sigma''} \in \rolR^\Inter\}=\{\rolR\in \lang{R}\mid \tuple{\sigma,\sigma'} \in \rolR^\Inter\}$ and $\{\conE\in
\parts(\kb^*)\mid \sigma''\in \conE^\Inter\}=\{\conE\in
\parts(\kb^*)\mid \sigma'\in \conE^\Inter\}$. W.l.o.g., $\sigma''$
does not get deleted in the whole construction procedure. Yet, then the $\Knter$-individual $\sigma''$ obviously contradicts $\sigma \in (\exists \rolU. \conD)^\Knter$.

\noindent
``$\Leftarrow$''\mbox{}\hspace{2mm}\mbox{}%
Assume the contrary, i.e., $\sigma \in (\forall\rolU.\conD)^\Inter$ but $\sigma \not\in (\forall \rolU. \conD)^\Knter$. The latter means that there is a $\Knter$-individual $\sigma'$ with $\tuple{\sigma,\sigma'}\in\rolU^{\Knter}$ and $\sigma'\not\in \conD^\Knter$. Because of the construction of $\Knter$ by pruning $\Inter$, this means also $\tuple{\sigma,\sigma'}\in\rolU^{\Inter}$ and $\sigma'\not\in \conD^\Inter$, ergo $\sigma \not\in (\forall\rolU.\conD)^\Inter$, contradicting the assumption.
\end{iteMize}

We proceed by showing that every role $\rolR$ with $\atmost{1}\rolR.\top \in \kb$ is functional in $\Knter$. Let $\sigma\in\Delta^\Knter$ and let $\sigma_1,\sigma_2$ be two $\rolR$-successors of $\sigma$. We consider two cases: First, assume that $\sigma_1=a_1^\Knter$ and $\sigma_2=a_2^\Knter$ for $a_1,a_2 \in \indnames$. Then, by construction of the unraveling we can derive that there must be an $a_3 \in \indnames$ with $\sigma=a_3^\Knter$. However, then, the DL-safe rule $R(x,y),R(x,z) \to y\approx z$ from $\Ef(\kb)$ ensures $\sigma_1=\sigma_2$. Next we consider the case that at least one of $\sigma_1,\sigma_2$ is unnamed. By Lemma \ref{lemma:identitas} and the point-wise correspondence between $\Inter$ and $\Knter$ shown in the previous part of the proof, two statements hold: First, for all $\conC \in\parts(\kb^*)$, we have that $\sigma_1\in \conC^\Knter$ iff $\sigma_2\in \conC^\Knter$. Second, for all $\rolS\in\rolnames$ we have that $\tuple{\sigma,\sigma_1}\in \rolS^\Knter$ iff $\tuple{\sigma,\sigma_2}\in \rolS^\Knter$. However, in the pruning process generating $\Knter$, exactly such duplicate occurrences are erased, leaving at most one $\rolR$-successor per $\sigma$. Thus we conclude $\sigma_1 = \sigma_2$.
This completes the proof that all axioms from $\kbrb$ are satisfied in $\Knter$.
\qed

Finally, we are ready to establish the equisatisfiability result also for this last transformation step.

\begin{theorem}
For any extended \ALCIFb{} knowledge base $\kb$, the \ALCIb{} knowledge base $\Ef(\kb)$ and $\kb$ are equisatisfiable.
\end{theorem}
\proof
Lemma~\ref{lemma:kb4kb5} ensures that every model of $\kb$ is also a model of $\Ef(\kb)$. Moreover, by Lemma~\ref{lemma:modelforf}, given a model $\Jnter$ for of $\Ef(\kb)$, any $\kb$-pruning of $\Jnter$'s unraveling (the existence of which is ensured by constructive definition) is a model of $\kb$. This finishes the proof.
\qed

\bigskip

Eventually, the results of this section can be composed to show how to transform an extended \SHIQb{} knowledge base $\kb$ into an equisatisfiable extended \ALCIb{} knowledge base by computing $\Eshq(\kbrb)\defeq\Ef\Ele\Eh\Ege\Es(\kbrb)$. \fancypicture{Figure~\ref{fig:DLTrafo} summarizes this procedure.
\begin{figure}[t]
\includegraphics[width=1\textwidth]{figures/executiveTrafo75}
\caption{Transformation from \SHIQ{} to \ALCIb{} TBoxes.\label{fig:DLTrafo}}
\end{figure}}
Moreover, as each of the single transformation steps is time polynomial, so is the overall procedure. Therefore, we are able to check the satisfiability of any extended \SHIQ{} knowledge base using the method presented in the previous sections, by first transforming it into \ALCIb{} and then checking.

This result is recorded in the below theorem, where we also exploit it to show an even stronger result about the correspondence between $\kbrb$ and $\Eshq(\kbrb)$.
\begin{theorem}\label{theo:completetrafo}
Let $\kb$ be an extended \SHIQb{} knowledge base. Then the following hold:
\begin{iteMize}{$\bullet$}
\item
$\kb$ and $\Eshq(\kb)$ are equisatisfiable,
\item
$\kb\models \conC(a)$ iff $\Eshq(\kb)\models {\conC}(a)$,
\item
$\kb\models \rolR(a,b)$ iff $\Eshq(\kb)\models {\rolR}(a,b)$, and
\item
$\kb\models a\approx b$ iff $\Eshq(\kb)\models a\approx b$,
\end{iteMize}
for any $a,b\in\indnames$, $\conC\in\connames$, and $\rolR\in\rolnames$.
\end{theorem}
\proof
Equisatisfiability follows from the fact that each of the transformations $\Ef,\Ele,\Eh,\Ege,\Es$ preserves satisfiability.
We then use the established equisatisfiability of $\kb$ and $\Eshq(\kb)$ to prove the other claims. Assume $\kb\models \conC(a)$. This means that the knowledge base $\kb'$ obtained by extending $\kb$ with the DL-safe rule $\conC(a)\to$ is unsatisfiable. Now we observe that $\Eshq(\kb')$ is obtained by extending $\Eshq(\kb)$ with $\conC(a)\to$. Since $\Eshq(\kb')$ is unsatisfiable, so is $\Eshq(\kb)$ extended with $\conC(a)\to$, and hence $\Eshq(\kb)\models {\conC}(a)$ as required. The other direction of the claim follows via a similar argumentation. The remaining cases are shown analogously.
\qed

\medskip

Consolidating all our results, we now can formulate our main theorem for checking satisfiability as well as entailment of positive and negative ground facts for extended \SHIQb{} knowledge bases.

\begin{theorem}
Let $\kbrb$ be an extended \SHIQb{} knowledge base and let $$\mathbb{P} \defeq \Prog(\Eshq(\kbrb)).$$ Then the following hold:
\begin{iteMize}{$\bullet$}
\item
$\kbrb$ is satisfiable iff $\mathbb{P}$ is,
\item
$\kbrb\models \conC(a)$ iff $\mathbb{P}\models S_{\conC}(a)$,
\item
$\kbrb\models \rolR(a,b)$ iff $\mathbb{P}\models S_{\rolR}(a,b)$, and
\item
$\kbrb\models a\approx b$ iff $\mathbb{P}\models a\approx b$,
\end{iteMize}
for any $a,b\in\indnames$, $\conC\in\connames$, and $\rolR\in\rolnames$.
\end{theorem}
\proof
Combine Theorem~\ref{theo:completetrafo} with Theorem~\ref{theo:aboxalcib}.
%
\qed


Note also that the above observation immediately allows us to add reasoning support for
\emph{DL-safe} conjunctive queries, i.e. conjunctive queries that assume all variables to range only over named individuals. It is easy to see that, as a minor extension, one could generally allow for concept expressions $\forall\rolR.\conA$ and $\exists\rolR.\conA$ in queries and rules, simply because $\Prog(\kbrb)$ represents these elements of $\parts(\FLAT(\mathscr{T}))$ as atomic symbols in disjunctive Datalog.

\section{Related Work}\label{sec:related}

Boolean constructors on roles have been investigated in the context
of both description and modal logics. \cite{borgida-dl-fol-96} used
them extensively for the definition of a DL that is equivalent to
the two-variable fragment of FOL.

It was shown by \cite{Hustadt:2000:IDD:647938.741235} that the DL
obtained by augmenting \ALC{} with full Boolean role constructors
($\mathcal{ALB}$) is decidable. \cite{lutz01complexity} established
\NExpTime-completeness of the standard reasoning tasks in this logic.
Restricting to only role negation \citep{lutz01complexity} or only role
conjunction \citep{Tobies:PhD} retains \ExpTime-completeness.
On the other hand, complexity does not increase beyond \NExpTime{} even when
allowing for
inverses, qualified number restrictions, and nominals. This was shown
by \cite{Tobies:PhD} via a polynomial translation of
$\mathcal{ALCOIQB}$ into $\mathcal{C}^2$, the two variable fragment
of first order logic with counting quantifiers, which in turn was
proven to be \NExpTime-complete by \cite{PH:C2complex}. Also the
description logic $\mathcal{ALBO}$ \citep{ALBO} falls in that range
of \NExpTime-complete DLs. 

On the contrary, it was also shown by \cite{Tobies:PhD} that
restricting to \emph{safe} Boolean role constructors keeps \ALC{}'s
reasoning complexity in \ExpTime{}, even when adding inverses and
qualified number restrictions (\ALCIQb).

For logics including modeling constructs that deal with role
composition like transitivity or -- more general -- complex role
inclusion axioms, results on complexities in the presence of Boolean
role constructors are more sparse.
\cite{DBLP:journals/jancl/LutzW05} show that \ALC{} can be extended
by negation and regular expressions on roles while keeping reasoning
within \ExpTime. Furthermore, \cite{DBLP:conf/aaai/CalvaneseEO07}
provided \ExpTime{} complexity for a similar logic that includes
inverses and qualified number restriction but reverts to safe
negation on roles.
The present work showed that reasoning remains in \ExpTime{} for
extended \SHIQb{} knowledge bases. 
Regarding DLs that combine nominals and role composition, it was shown that \emph{unsafe} Boolean role constructors
can be added to \SHOIQ{} and \SROIQ{} (resulting in DLs \SHOIQBs{} and \SROIQBs{}) 
without affecting their respective worst-case complexities of \NExpTime{} and \NExpExpTime{} \citep{RKH:Jelia-08}. The restriction to simple roles, on the other hand, is essential to retain decidability.
Furthermore, conjunctions of simple roles (which are trivially safe in the absence of role negation) can be added to tractable DLs of the \EL{} and DLP families without increasing their worst-case complexity \citep{RKH:Jelia-08}.

\medskip

Type-based reasoning techniques have been described sporadically in
the area of DLs but never been practically adopted.

\cite{Lutz-et-al-IANDC-05} use a particular kind of types, called
\emph{mosaics} for finite model reasoning.
\cite{DBLP:conf/wollic/EiterLOS09} use similar structures, called
\emph{knots} for query answering in the description logic \SHIQ{}.
Both notions show a similarity to the notion of \emph{(counting)
star types} used for reasoning in fragments of first order logic
(\citealp{PH:C2complex}), in that they do not only store information
about single domain individuals but also about all their direct
neighbors. As opposed to this, our notion of dominoes exhibits more
similarity to the notion of (non-counting) \emph{two-types} used in
first-order logic, e.g., by \cite{GradelOR97}; both notions encode
information related to pairs of domain individuals (rather than
whole neighborhoods).

The approach of constructing a canonical model (resp.~a sufficient
representation of it) in a downward manner (i.e., by pruning a larger
structure) shows some similarity to Pratt's type elimination
technique \citep{pratt}, originally used to decide
satisfiability of modal formulae.

Canonical models themselves have been a widely used notion in modal
logic \citep{popkorn,blackburn},
 however, due to the additional
expressive power of \ALCIb{} compared to standard modal logics like
K (being the modal logic counterpart of the description logic
$\mathcal{ALC}$), we had to substantially modify the notion of a
canonical model used there: in order to cope with number
restrictions, we use infinite tree models based on unravelings
whereas the canonical models in the mentioned approaches are
normally finite and obtained via filtrations.

Related in spirit (namely to use BDD-based reasoning for DL
reasoning tasks and to use a type elimination-like technique for
doing so) is the work presented by \cite{pansattlervardi}. However,
the established results as well as the approaches differ greatly
from ours: the authors establish a procedure for
deciding the satisfiability of $\mathcal{ALC}$ concepts in a setting
not allowing for general TBoxes, while our approach can check
satisfiability of \SHIQ{} (resp.~\ALCIb{}) knowledge bases
supporting general TBoxes, thereby generalizing the
results by \cite{pansattlervardi} significantly.

\medskip

The presented method for reasoning with DL-safe rules and
assertional data exhibits similarities to the algorithm underlying
the KAON2 reasoner
(\citealp{Motik:PhD,DBLP:journals/jar/HustadtMS07,DBLP:journals/iandc/HustadtMS08}).
In particular, pre-transformations are first applied to \SHIQ{}
knowledge bases, before a saturation procedure is applied to the
TBox part that results in a disjunctive Datalog program that can be
combined with the assertional part of the knowledge base. As in our
case, extensions with DL-safe rules and ground conjunctive queries
are possible. The processing presented here, however, is very
different from KAON2. Besides using OBDDs, it also employs Boolean
role constructors that admit an indirect encoding of number
restrictions. Moreover, as opposed to our approach, the
transformation in \cite{Motik:PhD} does not preserve \emph{all} ground
consequences: \SHIQ{} consequences of the form
$\rolR(a,b)$ with $\rolR$ being non-simple may not be entailed by
the created Datalog program. This shortcoming could, however, be
easily corrected along the lines of our approach.
On the other hand, the KAON2 transformation avoids the use of disjunctions
in Datalog for knowledge bases that are Horn (i.e., free of disjunctive information). Reasoning for Horn-\SHIQ{} can thus be done in \ExpTime{}, which is worst-case optimal \citep{KRH:HornDLs12}.
In contrast, our OBDD encoding requires disjunctive Datalog in all cases, leading to a \NExpTime{} procedure even for Horn-\SHIQ{}.

\section{Discussion}\label{sec:conc}

We have presented a new worst-case optimal reasoning algorithm for standard reasoning tasks for extended \SHIQb{} knowledge bases. The algorithm compiles \SHIQb{} terminologies into disjunctive Datalog programs, which are then combined with assertional information and DL-safe rules for satisfiability checking and (ground) query answering. To this end, OBDDs are used as a convenient intermediate data structure to process terminologies and are subsequently transformed into disjunctive Datalog programs that can naturally account for ABox data and DL-safe rules. The generation of disjunctive Datalog may require exponentially many computation steps, the cost of which depends on the concrete OBDD implementation at hand -- finding \emph{optimal OBDD encodings} is \NP-complete but heuristic approximations are often used in practice. Querying the disjunctive Datalog program then is co-\NP-complete w.r.t. the size of the ABox, so that the data complexity of the algorithm is worst-case optimal \citep{Motik:PhD}. Concerning combined complexity of testing the satisfiability of extended knowledge bases, the \ExpTime{} OBDD construction step dominates the subsequent disjunctive Datalog reasoning part, so the overall combined complexity of the algorithm is \ExpTime{} resulting in worst-case optimality for this case as well, given the \ExpTime-hardness of satisfiability checking in \SHIQb{}.

It is also worthwhile to briefly discuss the applicability of our method to knowledge bases featuring so-called \emph{complex role inclusion axioms} (RIAs). By means of  techniques described by \cite{kazakov08:sroiqcompl}, any (pure, that is, non-extended) \SRIQbs{} knowledge base can be transformed into an equisatsfiable \ALCHIQb{} knowledge base, however, like Motik's original transitivity elimination, this transformation does not preserve all ground consequences. Consequently, it is not satisfiability-preserving for extended \SRIQbs{} knowledge bases. Still, capitalizing on these RIA-removal techniques, our method provides a way for satisfiability checking for \SRIQbs{} knowledge bases without DL-safe rules that is worst-case optimal w.r.t.\ both combined and data complexity. We believe, however, that it would be not to hard a task to modify the transformation to even preserve ground consequences.


For future work, the algorithm needs to be evaluated in practice. A crude prototype implementation was used to generate the examples within this paper, and has shown to outperform tableaux reasoners in certain handcrafted cases, but more extensive evaluations with an optimized implementation on real-world ontologies are needed for a conclusive statement on the practical potential of this new reasoning strategy. It is also evident that redundancy elimination techniques are required to reduce the number of generated Datalog rules, which is also an important aspect of the KAON2 implementation.

Another avenue for future research is the extension of the approach to more modeling features such as role chain axioms and nominals -- significant revisions of the model-theoretic considerations are needed for these cases.

%

%
%
%
%

\section*{Acknowledgements}

This work was supported by the DFG project ExpresST: Expressive
Querying for Semantic Technologies and by the EPSRC grant HermiT:
Reasoning with Large Ontologies.

We thank Boris Motik and Uli Sattler for useful discussions on
related approaches as well as Giuseppe DeGiacomo and Birte Glimm for
hints on the origins of some techniques employed by us. We also
thank the anonymous reviewers for their very thorough scrutiny of an
earlier version of this article as well as for their comments and
questions which helped to make the article more comprehensible and
accurate.

\bibliographystyle{apalike}
\bibliography{references}

\end{document}